%% file: surrogate_with_partial.tex
\documentclass[11pt]{article}
\usepackage[margin=1in]{geometry}

\input{preamble}

\begin{document}
\def\spacingset#1{\renewcommand{\baselinestretch}%
{#1}\small\normalsize} \spacingset{1}

\newcommand{\blind}{0}

\if0\blind
{
  \title{\bf \ourtitle}
  \author{\ourauthors}
  \date{}
  \maketitle
} \fi

\if1\blind
{
  \bigskip
  \bigskip
  \bigskip
  \begin{center}
    {\LARGE\ourtitle}
\end{center}
  \medskip
} \fi

\bigskip
\begin{abstract}
    Gaussian process surrogates are a popular alternative to directly using computationally expensive simulation models.  When the simulation output consists of many responses, dimension-reduction techniques are often employed to construct these surrogates.  However, surrogate methods with dimension reduction generally rely on complete output training data.  This article proposes a new Gaussian process surrogate method that permits the use of partially observed output while remaining computationally efficient.  The new method involves the imputation of missing values and the adjustment of the covariance matrix used for Gaussian process inference.  The resulting surrogate represents the available responses, disregards the missing responses, and provides meaningful uncertainty quantification.  The proposed approach is shown to offer sharper inference than alternatives in a simulation study and a case study where an energy density functional model that frequently returns incomplete output is calibrated.
\end{abstract}

\noindent%
\vfill

\newpage
\spacingset{2} 

\section{Introduction} \label{sec:intro}
Computer simulations are used to understand and analyze systems where directly experimenting on the real system is difficult or infeasible.  
The output of a simulation model depends on a user-specified input configuration that defines the physical and controllable properties of the system.  
When a user simulates the system, also referred to as running the simulation, they receive outputs related to quantities of interest to the user.
Running a simulation can be computationally expensive; each run can cost thousands of core-hours, see for examples the simulation of storm surge \citep{plumlee2021high}, influenza spread \citep{Venkatramanan2021}, and nuclear dynamics \citep{phillips2021get}.
Because these simulations are expensive, it is often helpful to build an emulator, or ``surrogate,'' trained on simulation data to predict at unsimulated (i.e., out-of-sample) configurations.

Surrogate technology is often deployed for calibration \citep{kennedy2001bayesian}, where an input configuration is represented by a multidimensional parameter.
When run at a parameter, the simulation returns a high-dimensional output consisting of multiple responses collected on a set of fixed  locations.
While there are many variations of the exact type of inference \citep{tuo2015efficient,gramacy2015calibrating,plumlee2017bayesian}, the overall goal is to learn the parameters by aligning the physical observations with the simulation outputs using computational tools like Markov Chain Monte Carlo (MCMC). 
Because simulation runs are expensive, it is not possible to run the simulation the millions of times needed for MCMC. 
Instead,  the idea is to run the simulation model for a set of representative parameters and to collect the simulation output from each run; consequently, the surrogate is constructed as the prediction conditional on the simulation output data. 
Important to solving the calibration problem is that the surrogate produces a measure of  uncertainty in the surrogate's predictions.

The most prominent tool for building statistical surrogates involves Gaussian processes (GPs) \citep{santner2018design,gramacy2020surrogates}.
GPs offer both an accurate prediction and a measure of uncertainty.
The case of high-dimensional outputs leads to the computational intractability of many surrogate construction tools.
Two approaches have been proposed to remedy this computational challenge.  The first one employs a Kronecker structure of the covariance function of the GP, which assumes a separation between parameters and locations \citep{rougier2008efficient, hung2015analysis, guillas2018functional, marque2020efficient}.  The second approach employs a dimension-reduction step for the outputs before the building of surrogates \citep{bayarri2007computer, higdon2008computer, gu2020fast}.  This is in contrast to data-reduction procedures that seek to choose a smaller set of points to represent the entire parameter space (e.g., the selection of support points \citep{mak2018support}).
Examples of dimension-reduction procedures include extractions of principal components \citep{higdon2008computer, lawrence2017mira, gu2020fast}, wavelet coefficients \citep{bayarri2007computer}, and calibration-optimal bases \citep{salter2019uncertainty} from the simulation output data.  (\citet{salter2019uncertainty} require the knowledge of physical observations, in addition to the simulation output.)  These procedures require complete data, meaning for each run, the entire output has to be returned by the simulation. 
However, in seeking high levels in both performance and fidelity, modern simulation codes may produce partially observed outputs. \citet{hung2015analysis} have developed an EM algorithm to address the issue of partially observed output prior to constructing a surrogate, and their method is included for comparison in this article (see Section \ref{sec:experiments}).  In this article, we focus on extending the second approach to incorporate partially observed outputs.

The presence of partially observed output can be attributed to various causes of code failures.  
One cause is failure in parallel computing environments where separate computations are scattered over a large number of computing nodes.  
If a computing node experiences failure during calculations, only some of the calculations may be completed.  
Another cause is related to numerical calculations embedded in simulation codes.  
For example, in a simulation that involves solving a system of nonlinear equations, if the system corresponding to a response is inconsistent or particularly ill-conditioned, then no meaningful solution may be found numerically.  
Another cause comes from simulations where some responses in an output are undefined.
It is not always easy to identify a single underlying cause.  
Consider a parallel computing environment where a response is not returned because a simulation is terminated by the environment when it exceeds a time limit.
The response could be missing because the time limit was set too low, but it could also be missing because the numerical calculations within the simulation would never have terminated.
Regardless of the cause, the presence of partially observed outputs renders most surrogate methods inapplicable.
Partially observed outputs can be viewed in the context of data missingness.  
There are several classical categories of missingness mechanisms considered by statisticians:   missing completely at random (MCAR), missing at random (MAR), and missing not at random (MNAR).  
In practice the underlying causes of partially observed output are difficult to untangle because of the typically deterministic nature of computer code.
Failure in computing nodes seems as though it can be considered MCAR, but in many cases it is MAR because run time depends on the input and longer simulations are more likely to encounter a failure than a shorter one.
There are many examples of missingness in simulation codes, including
climate studies \citep{chang2014fast, ma2022computer}; fluid dynamics \citep{huang2020site};  nuclear physics \citep{bollapragada2021optimization}; and weather dynamics \citep{plumlee2021high}.

Despite the common occurrence of partially observed output, there are few methods available for high-dimensional surrogate construction in this setting. 
The na\"{i}ve approach would be to simply discard the dimension-reduction step entirely and treat each response separately.
In this case, one discards the missing responses and pushes forward with surrogate construction with the available responses. 
The na\"{i}ve approach includes locations as additional input dimensions.  This inclusion allows for the correlation structures among both parameters and locations to be modeled, similar to the inclusion of time index as an additional input \citep{bayarri2009predicting}.
However, this approach can easily exceed the limit of standard GP inference when the number of data points goes above several thousands.
The only recourse while staying with GP inference then leaves GP approximations, which are still significantly more expensive than dimension-reduction approaches and lead to decreased accuracy.
Another approach is to neglect correlations between locations and build an independent surrogate for each location.  When the output contains a small number of responses, building separate surrogates often suffices \citep{baker2022analyzing}.  However, as the number of responses increases, this approach is prohibited by its computational burden.  Simplistic imputation is yet another option, where one imputes the missing values and then builds a surrogate using dimension-reduction tools (e.g., \cite{plumlee2021high}).
The imputation process can be done in various ways besides prediction, such as a ``constant-liar'' imputation used in optimization (e.g., \cite{Chevalier2013}). 
For purposes of uncertainty quantification, these methods are dangerous as they will interpolate the imputed values with zero residual uncertainty.  
For example, if the entire output from a run is missing, the imputation approach would simply interpolate the imputed values instead of representing them as missing, which seems contrary to the goals of a surrogate to faithfully represent the predictive uncertainty.

In this article we propose a new GP surrogate method that operates well in the presence of partially observed simulation output.
The new method retains the computational efficiency found with dimension-reduction methods like those described in \citet{higdon2008computer}.  
However, in contrast to previous such approaches, the new method is not limited to complete data.  
This method involves two components: the imputation of missing values in the data using the principal components and an adjustment to the covariance matrix used in GP prediction. 
The adjustment to the covariance matrix ensures that one does not interpolate the imputed values at the missingness locations.
The resulting surrogate permits the use of data with partially observed output, and it has two appealing properties: (i) In the case where complete data is observed for a run, the resulting surrogate will interpolate those results and (ii) in the case where no data is observed for a run, the associated data row (even though imputed) will be ignored.  The new surrogate method demonstrates robustness empirically under various missingness mechanisms, and provides improved uncertainty quantification in calibration.

The organization of this article is as follows.  Section~\ref{sec:fayans} describes the Fayans energy density functional (EDF), the simulation model that motivates this surrogate development.  Section~\ref{sec:background} introduces the setting and notation employed, and reviews the principal component GP method for surrogate construction.  Section~\ref{sec:gpwm} details the imputation and covariance adjustment components in the new surrogate method alongside its properties.  Section~\ref{sec:experiments} details a numerical experiment to illustrate the surrogate performance under multiple missingness scenarios.  Section~\ref{sec:application_fayans} discusses the calibration of the Fayans EDF model.  Section~\ref{sec:conclusion} provides further discussions and closes the article.

\section{Fayans energy density functional}\label{sec:fayans}

The development of our proposed surrogate method is motivated by problems such as the calibration of the Fayans EDF \citep{fayans1998towards,fayans2000nuclear}. The development and refinement of EDF models have proven effective in understanding atomic nuclei \citep{reinhard2017toward}.  The EDF model investigated in this article was developed by S.\ A.\ Fayans and collaborators for describing ground-state properties of nuclei \citep{fayans1998towards,fayans2000nuclear}, and has since been used for nuclei predictions (e.g., see \citet{yu2003energy, reinhard2017toward}). 
\citet{bollapragada2021optimization} provides a full description of the numerical implementation under consideration in this article. The model takes a 13-dimensional parameter as input and outputs 198 responses.  
Each response corresponds to a spherical, ground-state, even-even nuclear configuration and an observable class.  A total of 72 nuclear configurations and 9 observable classes are considered, totalling to the 198 responses.  The list of responses can be found in Table 1 of \citet{bollapragada2021optimization}.  

The main relevant feature of this model is that it occasionally produces partial outputs, i.e., within one output, some out of the 198 responses are missing due to code failures, presenting a challenge in constructing a surrogate effectively.   
Missingness often occurs because an iterative method for solving equations is used in model evaluation, and failure is reported when the method fails to achieve a prescribed accuracy tolerance within an allotted internal iteration budget.   \citet{bollapragada2021optimization} outline the intricacies of possible failures and show that it is not likely the output is MCAR. 
However, no systematic missingness mechanism is proposed.
In previous analyses of another EDF model, $\approx 9\%$ of the model outputs were discarded due to such types of failures \citep{higdon2015bayesian}.

In Section~\ref{sec:application_fayans} we use a dataset of 500 parameters to construct a surrogate and calibrate the Fayans EDF model.  
An illustration of this is presented in Figure~\ref{fig:fayans_missing}, where out of 500 runs of the model, near $60\%~(359 / 500)$ have at least one missing value.
\begin{figure}[t]
    \centering
    \includegraphics[width=0.65\linewidth]{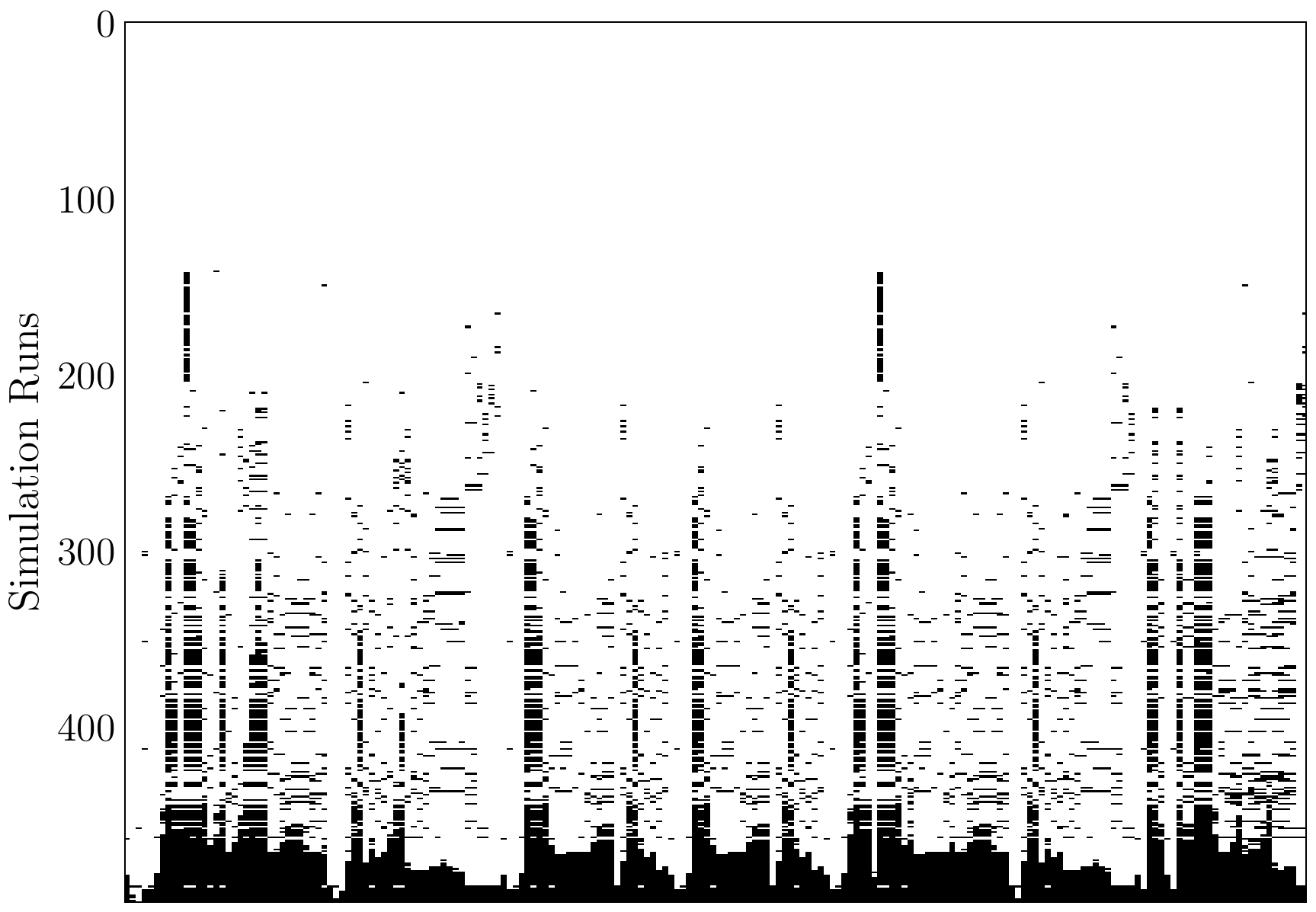}
    \caption{Illustration of partially observed output in our case study (Section \ref{sec:application_fayans}).  Each horizontal line (1--500) is a length 198 simulation output. A dark patch indicates where a response is missing.  The horizontal lines are sorted by number of missing values in the output.}
    \label{fig:fayans_missing}
\end{figure}

In calibrating the Fayans EDF model, \citet{bollapragada2021optimization} have previously employed a point minimization of the total mean-squared loss, or $\chi^2$ loss, with respect to the parameters.  Other works calibrated other EDFs in a similar manner, with minimization of the $\chi^2$ loss \citep{dobaczewski2005solution, kortelainen2010nuedf0, kortelainen2012nuedf1, kortelainen2014nuedf2, reinhard2017toward}.  Such approaches do not directly result in the uncertainty quantification  sought by nuclear physicists \citep{dobaczewski2014error}.
For some other EDFs, uncertainty quantification was performed under a Bayesian framework 
by \citet{higdon2015bayesian} and \citet{mcdonnell2015uncertainty},   but these works assume no missing data (or simply remove output with missing data). 

The existing literature does not contain a viable methodology to solve problems like this one.
One method is proposed in \citet{ma2022computer} to extract functional principal components, which originates from the analysis of partially observed longitudinal data.  The locations in \citet{ma2022computer} are irregularly spaced between output dimensions, producing partial, often sparse, output at each dimension.  In our problem, all responses are expected for each parameter, but missing values may arise for some responses because of various code failures.  
\citet{lawrence2017mira} have proposed to extract principal components using complete output data, and project the partially observed output data onto a subset of basis vectors to obtain the weight coefficients.  This treatment of partially observed output data is effective when the locations of missing data are regular.  Similarly, in 
\citet{gu2020fast}, the output data are modeled to follow a partition of complete and missing data blocks.  The complete data block is then used for principal component analysis following \citet{higdon2008computer}.  In their case, the complete data block contains a large proportion of the output data.  However, it is not applicable in our setting since our missing values are irregular. Thus the equivalent complete data block retains only a small proportion of the output data, resulting in inaccurate estimation of the principal components.  
Another method is proposed by \citet{hung2015analysis}, in which an expectation-maximization procedure is developed to tackle the issue of missing data, coupled with a separable correlation function that reduces the computational cost.  
The scale of their intended application is quite small, preventing direct application for our intended application.  
\citet{hung2015analysis} have studied a problem with $30$ runs and only a few runs with missing data; whereas in the Fayans EDF case, the number of simulation runs is $500$, and over half of them result in a partially observed output, recall Figure~\ref{fig:fayans_missing}. 
\citet{hung2015analysis} is revisited in the numerical experiments in Section~\ref{sec:experiments}.

\section{High-dimensional surrogates}\label{sec:background}
This section introduces the notations used in this article and reviews standard GP modeling and its principal component based extension to high-dimensional surrogate construction \citep{higdon2008computer}.

\subsection{Setting and notations}
We label the user-specified parameter $\bm \theta \in \mathbb{R}^d$.
Because surrogate inference is often deployed in calibration settings, $\bm \theta$ is referred to as a ``parameter.''  
We assume that for a simulation run at $\bm \theta$, the simulation output $\bm f(\bm\theta) \in \mathbb{R}^m$ is collected over a fixed set of locations, labeled as $\setX = (\bm x_1, \ldots, \bm x_m)$, and thus
$\bm f(\bm\theta) = (f(\bm \theta,\bm x_1), \ldots,  f(\bm\theta, \bm x_m) )^\T, $
meaning the simulation evaluated at $\bm \theta$ will produce an output vector $\bm f(\bm \theta)$ consisting of $m$ elements.
Each element in $\bm f(\bm \theta)$ is considered an individual response.
For example, in the Fayans EDF,  
the user-specified $\bm \theta$ represents a 13-dimensional parameter and 
$\bm x \in \setX$ defines a nuclear configuration that corresponds to a response where $m=198$ (see Section \ref{sec:fayans}).  

We presume that we have generated some parameters $(\bm\theta_1, \ldots, \bm\theta_n)$ from a designed computer experiment (such as a Latin Hypercube sample \citep{mckay1979comparison} or optimized variations thereof \citep{joseph2015maximum}), 
and the collection of simulation responses is arranged as the  $n \times m$ matrix $\bm F = (\bm f(\bm \theta_1)^\T, \ldots, \bm f(\bm \theta_n)^\T)^\T.$
Each row represents the simulation output for one parameter at all locations $\bm x \in \setX$.  
Each column represents the response collected at all parameters in the computer experiment.  
Since individual responses may be missing, we denote $N$ as the total number of available responses.
If the collection has complete data, $N = nm.$
For any matrix $\bm A$, we use $\bm A_{i, \cdot}$ to refer to the $i$th row, $\bm A_j$ to refer to the $j$th column of $\bm A$.
We use $\bm A_{\mathcal{I}\mathcal{J}}$ to denote a submatrix of $\bm A$ with entries from the set of row indices $\mathcal{I} \subseteq \{1,\ldots,n\}$ and column indices $\mathcal{J} \subseteq \{1,\ldots,m\}$.  Commas will be used to separate the row and column indices when the notation is ambiguous. 

\subsection{Gaussian process surrogates}
A surrogate is constructed to enable the prediction of output at an unsimulated parameter.
A good surrogate is designed to provide a predictive distribution, which can then be converted to point estimates and uncertainty around those estimates.
In this article, a surrogate provides a predictive distribution for $\bm f(\bm \theta)$ with a mean $\bm \mu(\bm \theta)$ and a covariance matrix $\bm \Sigma(\bm \theta)$.
GPs offer a path to do exactly this with a multivariate normal predictive distribution \citep{santner2018design, gramacy2020surrogates}. 
GP is a common choice to build a surrogate because of its flexibility and ability to interpolate.

\subsubsection{Univariate GP}
This section will review the basics behind GP inference.
Consider $g(\bm \theta)$ as a univariate function that takes $\bm\theta$ as its input.  
A GP model with mean function $\gamma(\cdot)$, scale parameter $\lambda$, and correlation function  $\rho(\cdot, \cdot)$ assumes that for any collection of $n$ scalar outputs (e.g., $\bm g = (g(\bm\theta_1), \ldots, g(\bm\theta_n))^\T$),  follows a multivariate normal distribution with mean $\bm\gamma = (\gamma(\bm\theta_1), \ldots, \gamma(\bm\theta_n))^\T$ and covariance matrix $\lambda \bm R$ where $\bm R = (\rho(\bm\theta_i, \bm\theta_j))_{i,j=1}^n$.  
Typical choices for $\rho$ include a squared exponential or a Mat\'{e}rn correlation function \citep{handcock1993bayesian}, but this choice does not impact the rationale of our method and is left open in this article.  Using the GP as a surrogate means that the predictive distribution given $\bm g$ at any test parameter $\bm\theta$ is a normal distribution with mean $\gamma(\bm\theta) + \bm r^\T(\bm\theta) \bm R^{-1} (\bm g - \bm \gamma)$ and variance $\lambda \left(\rho(\bm\theta, \bm\theta) - \bm r^\T(\bm\theta) \bm R^{-1} \bm r(\bm\theta)\right)$, where $\bm r(\bm\theta) = (\rho(\bm\theta, \bm\theta_1), \ldots, \rho(\bm\theta, \bm\theta_n))^\T$ is the correlation between $\bm\theta$ and $(\bm\theta_1, \ldots, \bm\theta_n)$.  The standard form of this inference thus requires $\bm R^{-1} (\bm g - \bm \gamma)$, which costs typically $O(n^3)$ operations to compute (though approximations do exist). 

The computational cost is a significant obstacle that requires consideration while conducting GP inference.
It is critical that the construction of the surrogate and predictions via the surrogate should be fast compared to running the simulation itself \citep{gramacy2020surrogates}.  In the high-dimensional case, the total number of responses is $N = nm$.
One could use off-the-shelf GP inference, where each response received from the simulation is individually modeled over both $\theta$ and $x$, i.e., the previously mentioned na{\"i}ve approach,  typically uses $O(N^3)$ computations to obtain the necessary matrix inverses and determinants. 
When $N$ is moderately large (e.g., $10^4$), na{\"i}vely using GP inference is difficult because of this computational burden.
For our case study, $m=198$ and hence $N\leq 10^4$ would constrain us to roughly $n$ less than about $50$.
The accuracy of a GP surrogate is tied to $n$, and the practical results for our case study found that at $n=50$ the surrogate is not sufficiently accurate. 
Researchers have considered this setting before and noted that the problem arises when the model output is multivariate with many responses.
A high-dimensional output happens when the number of responses $m$ gets large (often above $>20$).
There are a few choices for high-dimensional GP surrogates \citep{bayarri2007computer,higdon2008computer,conti2010bayesian},  with the majority of these methods seeking to reduce the required computations to $O(n^3)$, such that the computational burden effectively depends only on $n$.

\subsubsection{High-dimensional GP surrogates}
\citet{higdon2008computer} describe a powerful tool for building surrogates in high-dimensional output settings.
Examples of successful applications include nuclear physics \citep{higdon2015bayesian} and storm surge \citep{kyprioti2021improvements} modeling.
The method works by leveraging a low-rank representation for the $m$-dimensional output by applying a singular value decomposition to the matrix $\bm F$. Specifically, this method seeks $\bm \Phi$, where $\bm \Phi$ is an $m \times \kappa$ matrix defined by a set of $\kappa$ orthonormal (i.e., $\bm\Phi^\T\bm\Phi = \bm I$) $m$-dimensional basis vectors.  These are chosen such that for some reasonable $m$-dimensional centering vector $\bm c$, we have $\bm F - \bm 1_n \bm c^\T \approx (\bm F - \bm 1_n \bm c^\T) \bm \Phi \bm \Phi^\T$, meaning we can approximately recover our simulation outputs 
using only the $n \times \kappa$ matrix $\bm G = (\bm F - \bm 1_n \bm c^\T) \bm \Phi$.  
The centering vector is often chosen as the mean of each column in $\bm F$.
Then, the original data $\bm F$ are nearly recovered by $\bm 1_n \bm c^\T + \bm G \bm\Phi^\T$. 
The value of $\kappa$ is typically chosen to offer sufficient recovery of $\bm F$ from $\bm G$. 
This representation method is  especially effective when there is a strong relationship among the responses of the simulation output because then $\kappa$ can be made small (i.e., $\kappa\ll m$).
Suppose $\bm g(\bm \theta) = \bm \Phi^\T (\bm f(\bm \theta)-\bm c)$.  Then, the surrogate for the simulation is constructed as $\bm c + \bm \Phi \bm g(\bm \theta)$. Thus our goal shifts from building a surrogate on an $m$-dimensional output $\bm f(\bm\theta)$ to building a surrogate on a $\kappa$-dimensional output $\bm g(\bm\theta)$.

Our goal is now to use information in the matrix $\bm G$ to infer on the projected output $\bm g(\bm\theta)$ for any $\bm \theta$.  
Due to the orthogonal construction of $\bm \Phi$, each component of $\bm g(\cdot)= (g_1(\cdot),\ldots,g_{\kappa}(\cdot))^\mathsf{T}$ is modeled using an independent GP.
Each $g_k(\cdot)$ then follows
\begin{equation}
g_k(\cdot) \sim \text{GP}\left(0,\lambda_k \rho_k(\cdot,\cdot) \right), \label{eq:g_kmodel}
\end{equation}
where, for component $k$, $\lambda_k > 0$ is a scale parameter, $\rho_k(\cdot, \cdot)$ is a correlation function, and we have set the mean to be zero for ease of exposition.
With data column $\bm G_k$, under the GP model, prediction for $g_k(\bm\theta)$ follows a normal distribution with mean and variance given by
\begin{equation} \label{eqn:gk_gppred}
    \hat{\mu}_k(\bm \theta) = \bm r_k^\T(\bm \theta) \bm R_k^{-1} \bm G_k \text{ and } \sigma^2_k(\bm \theta) = \lambda_k\left(\rho_k(\bm\theta, \bm\theta) - \bm r_k^\T(\bm \theta)\bm R_k^{-1}\bm r_k(\bm \theta)\right),
\end{equation}
where
$ \bm r_k(\bm\theta) = (\rho_k(\bm \theta, \bm \theta_1), \ldots, \rho_k(\bm\theta, \bm\theta_n))^\T$ and $\bm R_k = (\rho_k(\bm \theta_i, \bm \theta_j))_{i, j=1}^n. $
From this,  the surrogate is then defined by
\begin{equation}\label{eqn:f_gppred}
    \bm \mu(\bm\theta) = \bm c + \bm \Phi (\hat{\mu}_1(\bm \theta), \ldots, \hat{\mu}_\kappa(\bm\theta))^\T \text{ and }  \bm \Sigma(\bm\theta) = \bm \Phi
    \begin{pmatrix}
        \sigma^2_1(\bm \theta) & \ldots & 0 \\
        \vdots & \ddots & \vdots \\
        0 & \ldots & \sigma^2_{\kappa}(\bm \theta) \\
    \end{pmatrix}
    \bm \Phi^\T.
\end{equation}
We refer to this method of constructing a surrogate as Principal Component Gaussian Process (PCGP).  
For PCGP methods, $\bm\Sigma(\bm\theta)$ is not of full rank when $\kappa < m$, thus the predictive distribution is often degenerate.  
However, similar to \citet{higdon2008computer}, when a surrogate is used to facilitate parameter calibration, 
$\bm\Sigma(\bm\theta)$  is summed with a (typically diagonal)
strictly positive definite observation error covariance matrix (see Section~\ref{sec:case_calib}).  The resulting sum of the matrices is then full rank and thus $\bm\Sigma(\bm\theta)$ not being strictly positive definite does not pose a problem.

The inferences in \eqref{eqn:f_gppred} are used to provide the mean and covariance of the surrogate for $\bm f(\bm \theta)$.
If there are partially observed simulation outputs, there are reasonable ways to approximate the principal component matrix $\bm \Phi$ and the centering vector $\bm c$ (see, e.g., \citet{roweis1997em, tipping1999probabilistic}).  The details of approximating $\bm \Phi$, an expectation-maximization algorithm inspired by \citet{roweis1997em}, are included in the Supplementary Material  \citep{chan2021supplement}.
However, partially observed simulation output in $\bm F$ means that $\bm G $ cannot be computed. 
Take the $i$th row of the data for example, if $\bm f(\bm \theta_i)^\T$ has one or more missing responses, then $\bm G_{i,\cdot}$ cannot be computed.  
When $m$ is large, it can often be the case that at least one response can be missing in a row, which leaves the rest of the data in the same row unusable without remedy.  
We then risk tossing out a large amount of data because of a few failures.  \citet{higdon2008computer} also commented on the necessity of complete data for the use of PCGP.  This leaves us at a crossroads.  If we use the PCGP approach, we cannot handle missing data.  If we do not adopt the PCGP approach, the surrogate computation may be impossible due to the explosive increase in required computations.  
This article explains how one can expand the PCGP surrogate methodology to handle missing data.

\section{Fast surrogates with missing data} \label{sec:gpwm}
We now introduce our method for constructing a surrogate for high-dimensional outputs
in the presence of missing data.  Specifically, we propose an imputation method for the missing observations during the computation of $\bm G$ alongside a novel adjustment of the covariance matrix.  The resulting method is computationally inexpensive and therefore the benefit of fast construction using PCGP is retained.    

Section~\ref{sec:pcgpwm_goals} describes and illustrates the desired properties of the proposed surrogate.  Sections~\ref{sec:pcgpwm_impute}~and~\ref{sec:pcgpwm_covadjust} detail the imputation and covariance adjustment procedures.  Section~\ref{sec:pcgpwm_thms} shows how
the proposed method delivers the desired surrogate properties.  Section~\ref{sec:pcgpwm_hyperparameter} describes the hyperparameter estimation used in the construction of the surrogate. Section~\ref{sec:pcgpwm_extraterm} provides additional justification for the choice of covariance adjustment.

\subsection{Desired surrogate properties} \label{sec:pcgpwm_goals}
In constructing this surrogate, we consider it necessary to have two desirable properties, the recovery of complete data rows and the ignorance of entirely missing data rows.
Meaning, we want to nearly interpolate to chosen precision where the output is complete  and disregard where the output is entirely missing.

The motivation for these desired properties can be explained through a thought experiment.
Let the collected data $\bm F$ be an $n \times m$ matrix, corresponding to simulation outputs at the parameters $(\bm\theta_1, \ldots, \bm\theta_n)$.  
Let $\bm F$ only have data available in $n_0 < n$ rows, and for these rows, the outputs are complete.  Let the $n_0 \times m$ matrix $\bm F^{(0)}$ denote the available data rows; without loss of generality, we take $\bm F^{(0)}_{i, \cdot} = \bm F_{i, \cdot}$ for all $i \leq n_0$. We assume that all of the remaining rows are entirely missing; this can occur, for example, when the computer running the simulations went down after completing $n_0$ rows. 
Suppose we separately construct two surrogates using a proposed method: one with data $\bm F$ and one with data $\bm F^{(0)}$.  
Let the surrogate with $\bm F$ be described  by $(\bm\mu(\bm\theta), \, \bm\Sigma(\bm\theta))$, and the surrogate with $\bm F^{(0)}$ be described by $(\bm\mu^{(0)}(\bm\theta), \, \bm\Sigma^{(0)}(\bm\theta))$. 
The recovery of complete data rows means that the surrogate should nearly interpolate for the rows with complete data; that is, $\bm\mu(\bm\theta_i) \approx \bm f(\bm\theta_i)$ and $\bm\mu^{(0)}(\bm\theta_i) \approx \bm f(\bm\theta_i)$ for all $i\le n_0$. 
The only interpolation error that should exist is due to the dimension reduction ($\kappa < m$) representation we have chosen.
The ignorance of entirely missing data rows means that the surrogate should not depend on any of the rows $\bm f(\bm \theta_i)^\T$ for any $i$ larger than $n_0$; that is, $\bm\mu(\bm\theta) = \bm\mu^{(0)}(\bm\theta)$ and $\bm\Sigma(\bm\theta) = \bm\Sigma^{(0)}(\bm\theta)$ for any $\bm\theta$.

Let us see how obvious approaches, namely the na{\"i}ve and simplistic imputation approaches described in the introduction, fare with respect to these goals.  
The na\"{i}ve method, where one treats each data point individually, would interpolate all observed points and ignore the remainder. 
Thus it meets our goals, but this method becomes computationally intractable in our settings because of the large $N$ problem.
The simplistic imputation approach would impute all rows when dealing with $\bm F$ and then interpolate that imputation.
This implies it would nearly interpolate all observed rows, but the predictions from the simplistic imputation approach using $\bm F^{(0)}$ and $\bm F$ are inconsistent.
This inconsistency implies that throwing out or keeping rows with fully missing output will give different predictions.

The method proposed in this article is able to nearly interpolate complete rows and ignore missing rows similar to the na\"{i}ve method, yet it remains computationally tractable.
While these criteria do not guarantee interpolation of outputs in the partially observed output case, it is our expectation and experience that by matching these two extremes, the predictions in the partially observed case offer notably better predictions than the simplistic imputation approach.
We later justify this with simulation experiments in Section~\ref{sec:experiments}.

\subsection{Gaussian process-based imputation} \label{sec:pcgpwm_impute}
We will assume the GPs modeling $g_k(\cdot)$ are stationary, and thus, without loss of generality, let $\rho_k(\bm\theta, \bm\theta) = 1$ for $k = 1,\ldots,\kappa$.  Furthermore, this section will assume the centering constant $\bm c=\bm 0$ for ease of exposition.  
If we consider the relationship that our output is (nearly) $\bm \Phi \bm g(\bm \theta)$, its covariance, following \eqref{eq:g_kmodel}, is of the form
\begin{equation} \label{eqn:cov_f}
    \bm \Phi \Cov(\bm g(\bm \theta)) \bm \Phi^\T =  \bm \Phi \bm \Lambda \bm \Phi^\T,
\end{equation}
where $\bm \Lambda$ is the diagonal matrix $\mathrm{diag}(\lambda_1, \ldots, \lambda_\kappa)$, due to the independence among the $\kappa$ GPs in \eqref{eq:g_kmodel}.
These values are presumed decided in advance.  For example, a reasonable choice deployed in our examples sets $\lambda_1, \ldots, \lambda_\kappa$ to the square of the corresponding singular values from the decomposition of $\bm F$.

The covariance matrix in \eqref{eqn:cov_f}, similar to the one in \eqref{eqn:f_gppred}, is not of full rank.  Therefore, it is difficult to use \eqref{eqn:cov_f} because the low-rank nature implies that one can have effectively ``fully observed'' $\bm f(\bm \theta)$ after observing only $\kappa$ entries.
Instead, we will use a full-rank extension of this matrix.
Pick $\varepsilon > 0$ such that $\lambda_k > \varepsilon$ for all $k \leq \kappa$.  This choice of $\varepsilon$ ensures that this covariance matrix extension for $ \bm f(\bm \theta) $, defined as
\begin{equation}\label{eq:Bcov}
    \bm B = \bm \Phi (\bm \Lambda - \varepsilon \bm I) \bm \Phi^\T + \varepsilon \bm I,
\end{equation}
is of full rank, and thus $\bm B$ and any principal submatrix of $\bm B$ are invertible. 

Using the covariance matrix in \eqref{eq:Bcov}, we impute the missing observations to build $\tilde{\bm\projobsdata}$, an $n 
\times \kappa$ matrix with complete entries, to replace $\bm G$.  

Let $\mathcal{J}(i) \subseteq \{1,\ldots, m\}$ be the set of column indices where data are not missing in $\bm f(\bm \theta_i)^\T$.  Then, let $\bm F_{i, \mathcal{J}(i)} = (f(\bm \theta_i, \bm x_j))_{j\in{\mathcal{J}(i)}}^\T$ be the row vector that contains the available data for parameter $\bm \theta_i$ and $\bm B_{\mathcal{J}(i),\cdot}, \bm B_{\mathcal{J}(i),\mathcal{J}(i)}$ be the submatrices corresponding to the indices.  
Because this section assumes a centering vector of zeros, 
\begin{equation*}
    \begin{pmatrix}
        \bm f(\bm \theta_i) \\ \bm F_{i, \mathcal{J}(i)}^\T 
    \end{pmatrix}
    \sim N \left( \begin{pmatrix}
        \bm 0 \\ \bm 0
    \end{pmatrix}, \begin{pmatrix}
        \bm B & \bm B^\T_{\mathcal{J}(i), \cdot}  \\
        \bm B_{\mathcal{J}(i), \cdot} & \bm B_{\mathcal{J}(i), \mathcal{J}(i)}
    \end{pmatrix} \right),
\end{equation*}
where this is a degenerate multivariate normal distribution. Subsequently the conditional mean and covariance of $\bm f(\bm \theta_i)$ given $\bm F_{i, \mathcal{J}(i)}^\T$, following standard normal theory (e.g.,  \citet{gelman2013bayesian}), are 
\begin{equation*}
    \mathbb{E}(\bm f(\bm \theta_i) | \bm F_{i,\mathcal{J}(i)}^\T) = \left(\bm F_{i,\mathcal{J}(i)} \bm B_{\mathcal{J}(i),\mathcal{J}(i)}^{-1} \bm B_{\mathcal{J}(i),\cdot}\right)^\T \quad \text{ and}
\end{equation*}
\begin{equation*}
    \mathrm{Cov}(\bm f(\bm \theta_i) | \bm F_{i, \mathcal{J}(i)}^\T) = \bm B - \bm B_{\mathcal{J}(i),\cdot}^\T \bm B_{\mathcal{J}(i),\mathcal{J}(i)}^{-1} \bm B_{\mathcal{J}(i),\cdot},
\end{equation*}
respectively. We propose then to infer $\bm G_{i, k}$ using the conditional quantities 
\begin{equation}\label{eqn:gtildemean}
    \tilde{\bm \projobsdata}_{i, k} = \bm F_{i,\mathcal{J}(i)} \bm B_{\mathcal{J}(i),\mathcal{J}(i)}^{-1} \bm B_{\mathcal{J}(i),\cdot} \bm \Phi_{\cdot, k} \quad \text{ and}
\end{equation} 
\begin{equation}
    u_{ik} = \bm \Phi_{\cdot, k}^\T \left(\bm B - \bm B_{\mathcal{J}(i),\cdot}^\T \bm B_{\mathcal{J}(i),\mathcal{J}(i)}^{-1} \bm B_{\mathcal{J}(i),\cdot} \right) \bm \Phi_{\cdot, k}. \label{eqn:gtildevar}
\end{equation}

This proposed inference is the key insight for our approach.  
While \eqref{eqn:gtildemean} is a reasonable use of our principal components for imputation, \eqref{eqn:gtildevar} is a valuable measure of goodness for the imputation, since the conditional variances reflect the imputation uncertainty.   
The quantities $u_{ik}$ from \eqref{eqn:gtildevar} are used to adjust our GP's covariance matrix to account for the uncertainty due to imputation.
For the suggested inferences to be practically useful, the inversion of the submatrix of $\bm B$ should be efficient.  The details of such inversion are provided in the Supplementary Material  \citep{chan2021supplement}.

 \subsection{Covariance adjustment} \label{sec:pcgpwm_covadjust}
The additional uncertainty from the imputation of missing data needs to be accounted for. 
We propose adjusting the covariance matrix by adding an extra term to $\bm R_k$, the original covariance matrix in \eqref{eqn:gk_gppred}, in order to model the increase in variance of prediction due to imputation. 
Defining the scaled predictive variances by $w_{ik} = u_{ik} / \lambda_k$,
the proposed adjusted covariance matrix is 
\begin{align} 
    \tilde{\bm R}_k = \bm R_k + \beta_k \mathrm{diag}(v_{1k}, \ldots, v_{nk}), \qquad \text{ where }
    v_{ik} = \min\left\{\eta,\frac{w_{ik}}{(1 - w_{ik})^\alpha}\right\}, \label{eqn:adjusted_cov}
\end{align}
where $\eta > 0$ is a large constant introduced to prevent infinite values from corrupting our linear algebra, and $\alpha>0$ and $\beta_k > 0$ are hyperparameters that affect the magnitude of the additional term. In particular, $\alpha$ controls the penalty for extreme missingness in an output and $\beta_k$ controls the  amplification of additional variances for component $k$.  In the case if we set $\beta_k = 0$, the variance added from imputation is not accounted for, and the imputed values are treated as observed and interpolated.

Using \eqref{eqn:gtildemean} on every row with missing data, $\tilde{\bm G}_k$ is the completed $k$th column of the low-dimensional data. Combining that with our covariance in \eqref{eqn:adjusted_cov} gives that  prediction for $g_k(\theta)$ with
\begin{equation} \label{eqn:tildemu_k}
    \tilde{\mu}_k(\bm\theta) = \bm r_k^\T(\bm \theta) \tilde{\bm R}_k^{-1} \tilde{\bm \projobsdata}_k \quad \text{and}
\end{equation}\begin{equation} \label{eqn:tildesigma_k}
    \tilde{\sigma}^2_k(\bm\theta) = \lambda_k \left(\rho_k(\bm\theta, \bm\theta) - \bm r_k^\T(\bm \theta)\tilde{\bm R}_k^{-1}\bm r_k(\bm \theta)\right),
\end{equation}
where $\bm v_k = (v_{1k}, \ldots, v_{nk})^\T.$ 
  The completed version of the surrogate provides a prediction for $\bm f(\bm \theta)$ that has a multivariate normal distribution with updated mean and covariance
\begin{equation} \label{eqn:rev_surrogate}
    \bm\mu(\bm\theta) = \bm\Phi (\tilde{\mu}_1(\bm\theta), \ldots, \tilde{\mu}_\kappa(\bm\theta))^\T, \text{ and } 
    \bm\Sigma(\bm\theta) = \bm\Phi
    \begin{pmatrix}
        \tilde{\sigma}^2_1(\bm \theta) & \ldots & 0 \\
        \vdots & \ddots & \vdots \\
        0 & \ldots & \tilde{\sigma}^2_{\kappa}(\bm \theta) \\
    \end{pmatrix}
    \bm \Phi^\T.
\end{equation}

Recall that for this section we set $\bm c=0$, thus these terms agree closely with the PCGP terms in \eqref{eqn:f_gppred}, where we have only modified the predictions of means and variances.  
If there is no missingness in our simulation output, the expressions in \eqref{eqn:f_gppred} and \eqref{eqn:rev_surrogate} will exactly agree.

\subsection{Properties of the proposed surrogate} \label{sec:pcgpwm_thms}
The following result
guarantees that if there is no missing data in the $i$th row (i.e., $\mathcal{J}(i) = \{1,\ldots,m\}$), then $\tilde{\bm \projobsdata}_{i,\cdot}$ recovers $\bm \projobsdata_{i,\cdot}$ without additional uncertainty.  
\begin{theorem}[Recovery of fully observed row] \label{thm:recovery}
    If $\mathcal{J}(i) = \{1,\ldots,m\}$, then $\tilde{\mu}_k(\bm \theta_i) = \bm f(\bm \theta_i)^\mathsf{T} \bm \Phi_{\cdot,k} $ and $\tilde{\sigma}_k^2(\bm \theta_i) = 0$ for  $k = 1,\ldots,\kappa$.
\end{theorem}
The proof is given in the Supplementary Material  \citep{chan2021supplement}.  On the other hand, notice that if there is no data in the $i$th row (i.e., $\mathcal{J}(i)$ is the empty set), then
\begin{equation*}    
u_{ik} = \bm \Phi_{\cdot, k}^\T \left( \bm \Phi (\bm \Lambda - \varepsilon \bm I) \bm \Phi^\T  + \varepsilon \bm I \right)  \bm \Phi_{\cdot, k}  = \lambda_k - \varepsilon + \varepsilon = \lambda_k. 
\end{equation*}
Our next result establishes that our method naturally ignores these rows of missing data.  
\begin{theorem}[Ignorance of missing rows] \label{thm:ignorefullmissingrows}
Define the $k$th surrogate component, where all rows with entirely missing data are excluded from construction, by $\tilde{\underline{\mu}}_k$ and $\tilde{\underline{\sigma}}^2_k$ from \eqref{eqn:tildemu_k} and \eqref{eqn:tildesigma_k}, respectively. 
If $\mathcal{J}(i) =\emptyset$, then for any $\alpha > 0$, $\beta_k > 0$, $\bm \theta$, and $k \in \{1,\ldots,\kappa\}$, we have that  $\tilde{\mu}_k(\bm\theta) - \tilde{\underline{\mu}}_k(\bm\theta) \rightarrow 0$ and
$\tilde{\sigma}^2_k(\bm\theta) - \tilde{\underline{\sigma}}^2_k(\bm\theta) \rightarrow 0$ as $\eta \rightarrow \infty$.
\end{theorem}
The proof is provided in the Supplementary Material  \citep{chan2021supplement}.  Although Theorem \ref{thm:ignorefullmissingrows} is  stated in a limit as our numerically stabilizing parameter $\eta$ goes to infinity, we note that in our deployments on real problems that setting $\eta$ to $10$ results in reasonable ignoring behavior.

\subsection{Estimation of hyperparameters in surrogate construction} \label{sec:pcgpwm_hyperparameter}
To construct the surrogate is to fit the hyperparameters for all $k = 1,\ldots, \kappa$ components.  Besides the covariance adjustment coefficient $\beta_k$, additional hyperparameters are implicitly included in the notation of the covariance function $\rho_k(\bm \theta, \bm\theta') = \rho_k(\bm \theta, \bm \theta'; \bm \psi_k),$
where $\bm\psi_k$ encapsulates the covariance function hyperparameters.  Within this section, we use the notation $\tilde{\bm R}_k(\beta_k, \bm\psi_k)$ to highlight the dependence of the adjusted covariance matrix $\tilde{\bm R}_k$ on both $\beta_k$ and $\bm\psi_k$.  The surrogate construction involves the optimization of the log-likelihood with respect to the hyperparameters, which can be simplified to
\begin{equation*} 
	\left(\hat{\beta}_k, \hat{\bm\psi}_k\right) = \arg\max_{\beta_k, \bm\psi_k} \Big\{- \frac{1}{2} \log |\tilde{\bm R}_k(\beta_k, \bm \psi_k)|  - \frac{1}{2 \lambda_k} \tilde{\bm \projobsdata}_k ^\T \tilde{\bm R}_k(\beta_k, \bm \psi_k)^{-1} \tilde{\bm \projobsdata}_k\Big\},
\end{equation*}
where $\lvert \bm A \rvert$ returns the determinant of $\bm A$.
The evaluation of the log-likelihood involves the decomposition of the covariance matrix that costs $O(n^3)$ operations, where $n$ is the number of parameters.  In our deployment, we use a gradient-based L-BFGS optimization solver to approximate the maximum likelihood estimate.

\subsection{Rationale for the extra term} \label{sec:pcgpwm_extraterm}
We now provide justification for the choice of the extra term in \eqref{eqn:adjusted_cov} alongside hyperparameter $\alpha$.  While $w_{ik}$ can be used directly as the added variance terms, it is insufficient to represent the intuition that (i) if the row has complete data, no adjustments should be made, and (ii) if no data are observed in the row, the added term should be infinite, because we have no information about the row. 
In the case where $\alpha=0$, $v_{ik} = w_{ik}$ does not achieve infinity when no data are observed.  
When $\alpha>0$, the added term $v_{ik} = 0$ when $w_{ik} = 0$, and $v_{ik} = \infty$ when $w_{ik} = 1$.  Moreover, the magnitude of $\alpha$ controls the rate of increase of variance $v_{ik}$ as $w_{ik}$ increases. The variance term $v_{ik}$ represents the uncertainty introduced from observing only at the index set $\mathcal{J}(i)$ for row $i$.  As $\alpha$ increases, $v_{ik}$ increases with a faster rate to infinity as $w_{ik}$ approaches 1.  
The effect of $\alpha$ is illustrated in Figure ~\ref{fig:alpha_v}.  
The hyperparameter $\beta_k$ is introduced to control the inflation (or deflation) of additional variances for the $k$th surrogate component.  The $\beta_k$ value can vary across the components.  When $\beta_k = 0$, the additional variances have no effect to the surrogate, as if there is no uncertainty due to data missingness.  When $\beta_k $ is infinite, the overwhelming additional variances will cause any row with any amount of missing data to be ignored.  When $\beta_k = 1$, no adjustment is made.  The hyperparameters $\beta_k, k=1,\ldots,\kappa,$ are estimated in the surrogate construction.  
\begin{figure}[t]
    \centering
    \includegraphics[width=0.5\linewidth]{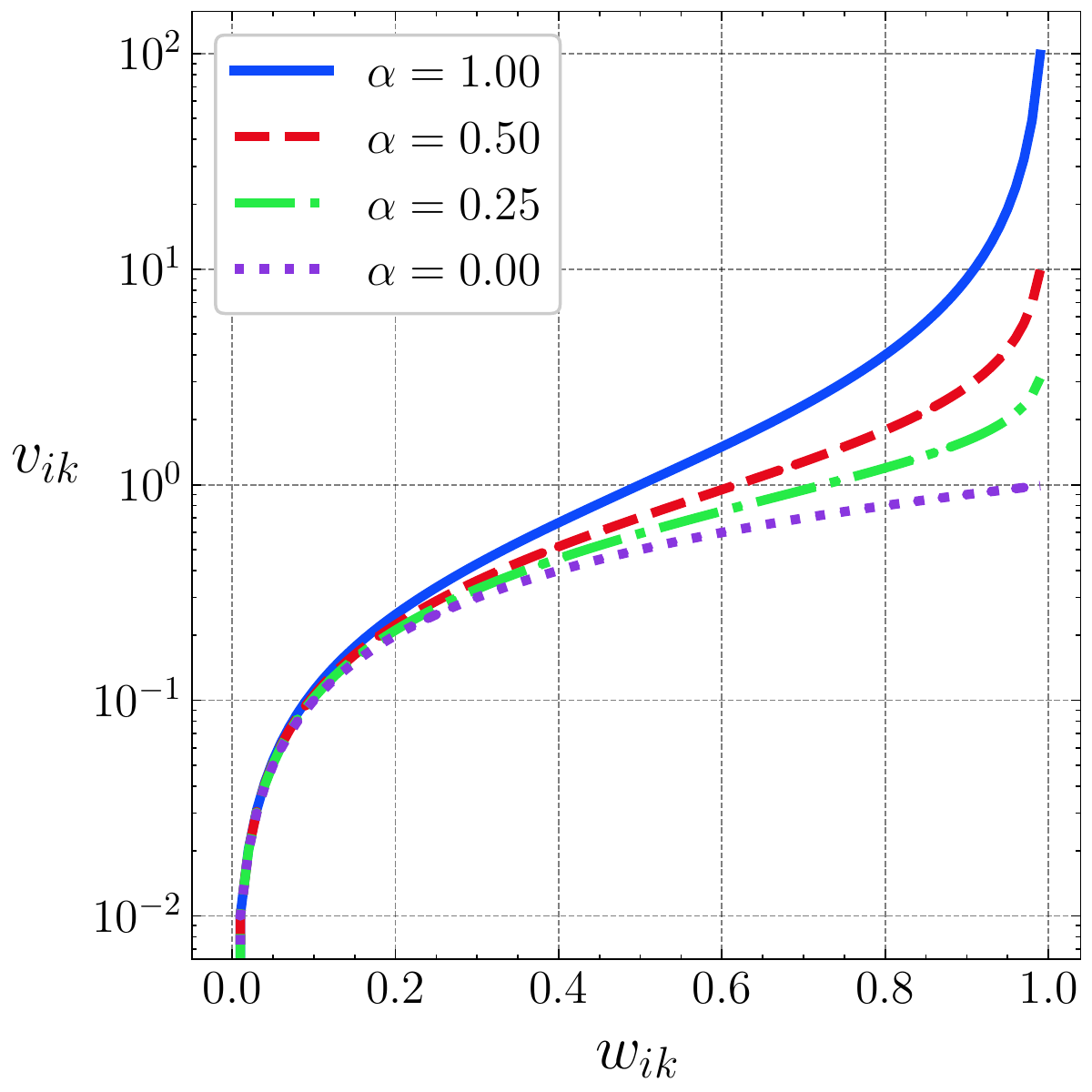}
    \caption{Illustrations of the effect of hyperparameter $\alpha$ in the variance term $v_{ik}.$ Except for $\alpha=0$, $v_{ik}$ approaches infinity as $w_{ik}$ approaches 1.}
    \label{fig:alpha_v}
\end{figure}

A numerical experiment is presented in Supplementary Material \citep{chan2021supplement} that investigates the choice of $\alpha$ and if $\beta_k$ should be either equal to $1$ or optimized as a hyperparameter. 
We found the best behavior when the $\beta_k$ are optimized and found $\alpha > 0.5$ could negatively impact the prediction accuracy.  
We suggest $\alpha = 0.3$ as a default choice.  

\section{Numerical experiments} \label{sec:experiments}
We now present the numerical experiment for evaluating the performance of our proposed surrogate method.  

\subsection{Setup of numerical experiments}
We compare the proposed method with other methods that employ common strategies for dealing with missing data.  A total of six methods are considered. 
We label our method as ``PCGPwM,'' for PCGP with Missingness.  Among the other methods, the first method (``GP-OM'', standing for GP-Omit-Missingness) is to omit the missing data, which is the na{\"i}ve approach that treats the $N$ responses individually and discards the missing points.  The second method (``colGP'') is to treat locations as independent and construct a GP for each of $m$ locations.  \citet{gu2016parallel} refer to colGP as the ``Many Single'' emulator approach and have commented on its computational intractability when $m$ is large.   \citet{gu2016parallel} resolves this computation issue by aligning common Gaussian process parameters across locations, but this fix is not needed here as our $m$ is relatively small.  We consider the expectation-maximization-based method by \citet{hung2015analysis} as the third method (``EMGP'').  
The next two methods involve the simplistic imputation approach with two different imputation procedures. 
The two imputation procedures considered are the $k$-nearest neighbor method (see, e.g., \citet{altman1992introduction}), and Bayesian ridge regression \citep{tipping2001sparse}.  Since the PCGP method is used on the imputed data, these two methods are termed ``PCGP-kNN'' and ``PCGP-BR,'' respectively.
In addition to these methods, we compare PCGPwM with a baseline method where the principal components are assumed known.  This comparison is reported in the Supplementary Materials \citep{chan2021supplement}.

GP-OM treats the missing data as if they were not present.  We discard any response that is missing.  After omitting the missing data, the remaining data do not necessarily form a data matrix.  For example, one row may have no missing data, while another row may have one or more entries removed. 
Instead, the remaining data are stacked into a column and a single GP is constructed as a surrogate.  Due to the stacking of data, this method becomes intractable as the computational complexity is $O(N^3)$, with $N$ being the number of available data points.  The colGP method treats each $\bm x\in \setX$ as independent locations and constructs $m$ GPs.  If $m$ is large, computational costs may not permit the use of colGP simply due to the number of GPs to be constructed.  The EMGP method estimates the missing values via an EM algorithm and uses a separable covariance matrix to speed up its GP construction.
For each of the simplistic imputations, the missing data are first imputed and PCGP is then applied to construct the surrogate.  The $k$-nearest neighbor method imputes the missing value by averaging the $k$ closest available data points, using the Euclidean distance between couples $(\bm\theta, \bm x)$.  The Bayesian ridge regression fits a regression model and imputes the missing values with the predictions of the model.  In the implementation of the imputation methods using \texttt{scikit-learn} \citep{scikit-learn}, missing values in each data column are imputed in a round-robin fashion for 10 iterations, and the result from the final iteration is returned.  See the package documentation on \texttt{IterativeImputer} for details.  

We use four functions as test examples, namely the borehole function, piston function, wingweight function, and OTLcircuit function.  These are common test functions for emulation and uncertainty quantification, located at \url{http://www.sfu.ca/~ssurjano/index.html}.  The input variables of each function are partitioned into $(\bm\theta, \bm x)$. 
Our proposed method does not address the source of missingness in the data. To test the robustness of the method, we generate partially observed output with the three missing mechanisms mentioned in Section \ref{sec:intro}.  The missing mechanisms are MCAR, MNAR, and MAR.  Under MCAR, any response has the same probability of missing.  Under MNAR, the response missingness depends on unobserved quantities, for example, the value of the response. Under MAR, the response missingness only depends on quantities that are observed, for example, the parameter and location in our case.  For each of these missingness mechanisms, the percentage of missingness is specified.  The methods are tested at 1\%, 5\%, and 25\% missingness levels.  With MNAR, values of the borehole and the wingweight functions are missing when the evaluated function value exceeds a threshold;  and the values of the OTLcircuit and the piston functions are missing according to the probability determined by a logistic model over the fixed locations.  With MAR, The missingness is randomly assigned over a subset of the fixed locations. The test functions and the details of the MNAR and MAR generation are included in the Supplementary Material \citep{chan2021supplement}.

The size of the output, $m$, is set as 15.  The set of locations, $(\bm x_1, \ldots, \bm x_m)$, This is determined using $\bm x$'s that are uniformly sampled from their respective ranges.  The number of parameters, $n$, are set to take values $n=50, 100, 250, 1000,$ and $ 2500.$  The parameters are sampled using a Latin Hypercube design \citep{santner2018design} in the unit hypercube, and then scaled to their respective ranges.  Denoting a missingness scenario to be the missing mechanism and the missing fraction combined, nine scenarios are considered: (MCAR, 1\%), (MCAR, 5\%), (MCAR, 25\%), (MNAR, 1\%), (MNAR, 5\%), (MNAR, 25\%), (MAR, 1\%), (MAR, 5\%), and (MAR, 25\%).  We construct a surrogate using each method, by supplying it the output of the model, for each combination of $n$ and missingness scenario.  To efficiently compare the methods, we fix the locations, the parameters, and the set of missing values across the four methods in one replication.  Each experiment is run for 20 replications for all test functions.  Each replication is given 1 hour of run time and is canceled if the surrogate construction and prediction takes longer than that.

All methods are implemented through the Python package \texttt{surmise} \citep{surmise2021}, a modular package that interfaces different statistical emulation and calibration methods.  Our proposed method is implemented under the name \texttt{PCGPwM}.  For GP-OM, \texttt{GPy} \citep{gpy2014} is used to construct the surrogate.  The colGP method is implemented under the name \texttt{colGP}.  The EMGP method is implemented under the name \texttt{EMGP}. For PCGP-kNN and PCGP-BR, first the imputations are performed with \texttt{scikit-learn} \citep{scikit-learn}, then the completed data are supplied to the \texttt{PCGP} method to construct a surrogate.  The simplistic imputation approach is implemented under the name \texttt{PCGPwImpute} with an option of which imputation method to use. 

\subsection{Results of comparison experiments} 
This section will discuss representative results for the (MNAR, 5\%) missingness scenario.  For a full description of the results of our simulation, we direct the reader to the Supplementary Material \citep{chan2021supplement}.  

All methods are competitive in computation, except GP-OM, which costs roughly 30 times as long for all $n$. 
 Na{\"i}vely omitting missing values in the data destroys the regular structure, and further leads to tractability issues.  In addition, EMGP does not complete at the largest data size within the time limit.  The remaining constructions are completed under the time limit.  Note that although the colGP method completes all surrogate constructions within time limit here with $m=15$, its computation time may be prohibitive as $m$ is too large, further explored in Section~\ref{sec:colgphighm}.

The quality of the surrogate methods is measured by the root mean squared error (RMSE), the coverage probability of the 90\% prediction interval (90\% coverage), and the width of the same interval (90\% width).  While RMSE concerns the predictive accuracy of the mean, 90\% coverage and 90\% width are empirical measures of the quality of the surrogate's uncertainty quantification.  The measures are evaluated against holdout simulation runs, and any missing values in the runs are excluded from evaluation.  

Figure \ref{fig:rmses5structured} shows the RMSEs for the surrogate methods being compared.  
For the borehole and the wingweight functions, the RMSE decreases for all methods except the simplistic imputation methods as $N$ increases; whereas for the piston and the OTLcircuit functions, the RMSEs decrease for all methods. GP-OM shows to be generally not competitive in both its computation time and predictive accuracy.  The simplistic imputation methods, PCGP-kNN and PCGP-BR, fail to circumvent the issue of missing values, especially when the missingness is MNAR.  Since the simplistic imputation approach relies on the availability of close neighbors to the missing points, if the function values are never observed within a certain region, the missing values would be imputed with far-away values of little relevant information. As a result, they result in poor predictions. The EMGP method performs comparably with the simplistic imputation methods, sometimes better, in the case of the wingweight function.  The colGP method performs well across all methods in predictive accuracy, especially in the piston and OTLcircuit test functions.  The continuous improvement with colGP as $N$ increases shows a potential drawback in using principal-component methods for dimension reduction.
Similar conclusions are drawn for (MNAR, 1\%) and (MNAR, 25\%).  We find that in the case of MCAR and MAR, the accuracy of all methods improves as $N$ grows, where PCGPwM performs better than all methods except for colGP.

\begin{figure}[t]
    \centering
    \includegraphics[width=0.95\textwidth]{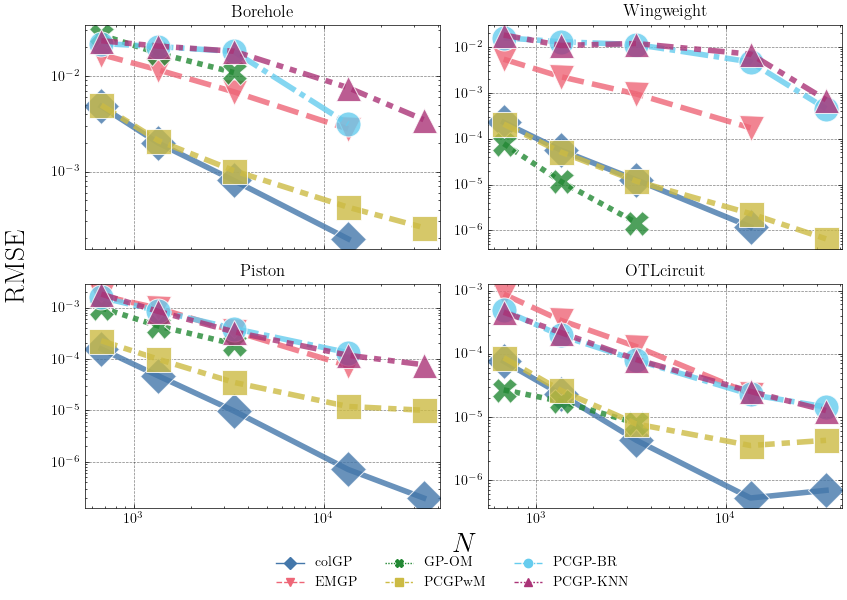}
    \caption{Comparison (log-log scale) of prediction accuracy of surrogate methods for (MNAR, 5\%) scenario.}
    \label{fig:rmses5structured}
\end{figure}

The 90\% coverage records how often the simulation response is contained in the interval produced by the surrogate. 
The quality of the surrogate is measured by comparing the 90\% empirical coverage with the nominal level (in this case, 90\%).  A coverage close to the nominal level with a narrow width is an indication of a good surrogate.
GP-OM results in overcoverage in the borehole function and significantly undercoverage in the remaining functions.  By investigating the interval widths, we observe that GP-OM produces a prediction interval too wide in the borehole and too narrow in the others.  The EMGP method attains adequate coverage in the piston and wingweight functions, while undercovers in the borehole and wingweight function.  As the data size grows, the coverage for the wingweight function increases and attains the prescribed level.
The simplistic imputation methods, PCGP-kNN and PCGP-BR, exhibit overcoverages in the borehole and wingweight functions, as a result of wide intervals. 
PCGPwM achieves the prescribed coverages while being able to provide sharper predictions, indicated by narrower intervals.
In the MCAR and MAR scenarios, both the coverage and the prediction interval widths improve as $N$ grows for all methods.

Overall, the proposed PCGPwM method performs generally well in terms of RMSEs and is robust to different types of data missingness.  The method also preserves the efficiency of surrogate construction found in dimension-reduction methods such as PCGP \citep{higdon2008computer}.

\subsection{Additional comparisons of PCGPwM and colGP at higher output dimension} \label{sec:colgphighm}
The predictive performance of colGP is laudible but it may be slow to construct at higher output dimensions. 
To further compare PCGPwM and colGP in terms of computational cost, we have conducted an additional experiment.  We have increased the output dimension to $m = 200$ and focused on the performance at larger $n={1000}, $ and ${2500}$ with the borehole function.  A maximum of 4 hours is permitted for each surrogate construction.  Table \ref{tab:compare_largem} reports the respective construction time, RMSE, 90\% coverage, and 90\% width.  PCGPwM achieves a higher RMSE than colGP, but maintains the right coverage with a narrower width.  At $n=2500$, the construction time of colGP exceeds the allowed time limit with higher dimensions and the computation is aborted.  According to the time scaling, the construction time would have taken colGP an estimated 18 hours. 

\begin{table}[h]
    \centering
    \begin{tabular}{c|cc|cc|cc|cc}
        \hline
       & \multicolumn{2}{c|}{Construction time (s)}   & \multicolumn{2}{c|}{RMSE ($\times 10^{-4}$)} & \multicolumn{2}{c|}{90\% coverage} & \multicolumn{2}{c}{90\% width}   \\
       \hline
       $n$  & PCGPwM & colGP & PCGPwM & colGP & PCGPwM & colGP & PCGPwM & colGP \\
       \hline
       \num{1000} & $\num{640}$ & $\num{4225}$ & \num{3.9} & \num{2.5} & 0.912 & 0.962 & 0.764 & 0.920 \\
       \num{2500} & $\num{9680}$ & -- & $\num{0.62}$ & -- & 0.981 & -- & 0.571 & --\\
       \hline
    \end{tabular}
    \caption{Construction times (in seconds) and predictive accuracies of PCGPwM and colGP at output dimension $m=200$.}
    \label{tab:compare_largem}
\end{table}

\section{Case study: Calibrating the Fayans energy density functional} \label{sec:application_fayans}
The section describes the case study mentioned in Section \ref{sec:fayans} that relies on the described surrogate method to conduct calibration.
We review a surrogate-based calibration framework in Section~\ref{sec:case_calib}.  
We present the results of applying the proposed method to the case of calibrating the Fayans EDF in Section~\ref{sec:case_apply}. 

\subsection{Bayesian calibration with a surrogate} \label{sec:case_calib}
Suppose $\bm y = (y_1, \ldots, y_m)^\T$ is a set of observations from the physical system that the simulation $\bm f(\bm \theta)$ is representing.  The differences between the observations and the simulation responses are assumed to follow a multivariate normal distribution with zero mean and covariance $\bm W$.  
This follows the canonical \citet{kennedy2001bayesian} framework with the assumption that a systematic bias is negligible.  The \citet{kennedy2001bayesian} framework has generated much research interest; for examples, see \citet{higdon2004combining,williams2006combining, bayarri2007computer, higdon2008computer,brynjarsdottir2014learning}.

Let $\pi$ denote a probability density and $\pi(\bm \theta|\bm y)$ be the conditional probability density of $\bm\theta$ given $\bm y$.  The purpose for calibration is to infer about the parameter $\bm\theta$.  In the Bayesian setting, we are interested in the quantity $\pi(\bm \theta|\bm y)$, which is the posterior density of $\bm\theta$.  By Bayes rule, given a prior density $\pi(\bm \theta)$, we have
$\pi(\bm\theta | \bm y) \propto \pi(\bm y | \bm \theta)\pi(\bm\theta),$
where $\propto$ denotes equality up to a constant multiplier.  Given the distribution of the differences, the expression is expanded to be
\begin{equation*}
    \pi(\bm\theta|\bm y) \propto \lvert \bm W \rvert^{-\frac{1}{2}} \exp\left(-\frac{1}{2}(\bm y - \bm f(\bm \theta))^\T \bm W^{-1}(\bm y - \bm f(\bm \theta))\right)\pi(\bm \theta),
\end{equation*}
where $\bm W $ is the covariance matrix of the observation error.  Often in a physical experiment, this covariance matrix is diagonal, as in the Fayans EDF model.

Since missing data may be present, $\bm f(\bm \theta)$ refers to the hypothetical output responses at a given parameter.  When a surrogate, defined by $\bm\mu(\bm\theta)$ and $\bm\Sigma(\bm\theta)$, constructed with data $\bm F$ is used in place of the simulation model, which is a common strategy in the calibration literature, the posterior density is revised as 
\begin{equation}
    \pi(\bm \theta|\bm y, \bm \obsdata) \propto \lvert \bm W + \bm \Sigma(\bm \theta) \rvert^{-\frac{1}{2}} \exp\left(-\frac{1}{2}(\bm y - \bm \mu(\bm \theta))^\T \left(\bm W + \bm \Sigma(\bm \theta)\right)^{-1}(\bm y - \bm \mu(\bm \theta))\right) \pi(\bm \theta), \label{eq:surrogatepost}
\end{equation}  
by considering the joint distribution of $\bm y$ and $\bm F$ being again multivariate normal.
For a fixed simulation model sample $\bm F$ (with missing data), this expression can be used to draw from the posterior using MCMC methods.  

\subsection{Surrogate construction and calibration for the Fayans EDF} \label{sec:case_apply}
In this section we apply the previously described surrogate methods for the calibration of the Fayans EDF model to provide uncertainty estimates for the parameter, using the knowledge of point minimizers obtained from the previous study.  
The first method is the proposed PCGPwM method.  The second uses the simplistic imputation approach with PCGP-kNN.  The third uses the colGP method, and the last method is simply to only use data rows with complete data.
GP-OM discussed in the preceding section is not usable in this setting due to the data size causing it to be computationally infeasible: we estimate that the surrogate construction alone would take more than $30$ days, before any calibration can be performed.  The EMGP method is not applicable either, due to its computational instability.  In the original paper by \citet{hung2015analysis}, an isotropic correlation function is used for the categorical variable, which has two levels.  In the Fayans EDF model, the categorical variable has nine levels.  To adopt the EMGP method, the correlation is chosen to be 1 if the categorical variable is the same, and 0 if not.  However, the hyperparameter estimation for the EMGP consistently fails to converge for this application.

We construct the surrogates using simulation outputs of $500$ well-spaced parameters.  All parameters reside close to a local minimizer of the $\chi^2$ loss from \citet{bollapragada2021optimization}.  
The parameter space is previously scaled to a unit hypercube such that the dimensions are comparable.  To recover the unscaled parameter the centroid of the unscaled hypercube and the scale for each dimension are provided in Table~\ref{tbl:fayans_param_space} in the Supplementary Material  \citep{chan2021supplement}, which reproduces Table~5 of \citet{bollapragada2021optimization}.  

The simulation outputs are partially observed, where approximately $10\%$ of the 99,000 responses are missing.  
Figure~\ref{fig:fayans_missing}, as seen in the introduction, shows the missing value pattern in the sampled outputs arranged by increasing number of failures in parameters.  Only $141$ rows have complete data.  

To calibrate the Fayans EDF, the constructed surrogates are then supplied to the calibrator module in \texttt{surmise}.  The chosen prior is the $\mathrm{Beta}(2,2)$ distribution over each dimension of the scaled parameter, representing the stability region studied in \citet{bollapragada2021optimization}.
The choice $\mathrm{Beta}(2, 2)$ reflects the understanding that the parameters closer to the centroid of the scaled hypercube are more plausible.  The prior density outside the scaled hypercube is zero, which reflects the boundary of stability region defined by the lower and upper bounds.  
The prior is then
\[\pi(\bm\theta) 
\propto \prod_{l=1}^d \theta_{l} \left(1 - \theta_{l}\right), \bm\theta \in [0, 1]^d, \]
where $\theta_l$ is the $l$th element of $\bm \theta$.
With $\pi(\bm\theta)$ specified, the posterior is then given by \eqref{eq:surrogatepost}.
Samples are drawn from the posterior by the Langevin Monte Carlo method \citep{roberts2002langevin},  an MCMC method that utilizes gradient information at the current iterate.  In addition, the sampling method is strengthened by incorporating parallel tempering \citep{geyer1991markov, gelman2013bayesian}.  The method is implemented in the utility module of \texttt{surmise} under the name \texttt{PTLMC}.   We note that the closed-form nature of our surrogate allows for easy deployment of gradient-based approaches.

To investigate the utility of the surrogate-based inference using these methods, we compare the posterior distributions.  Since the considered parameter space is scaled around the centroid of the hypercube, we expect the posterior means to be close.  This is verified in our computation that the posterior means estimated using all surrogate methods considered are found to be close to each other.  We are then concerned with the precision of the posterior, the idea being that a more precise surrogate in this setting should lead to a narrower posterior on the parameters.
This can be seen in expression \eqref{eq:surrogatepost}, where a decrease to the surrogate covariance matrix $\bm\Sigma(\bm\theta)$ yields a more concentrated posterior, given a fixed $\bm\mu(\bm\theta)$.  The precision of the surrogate is measured by the widths of the intervals between the $5\%$ and $95\%$ quantiles (called the $90\%$ width) for each parameter relative to the upper and lower bounds.
Table~\ref{tbl:posttheta_ci} contains the $90\%$ widths relatively scaled from all surrogates.
The credible intervals from all the surrogates shrink compared with the $\mathrm{Beta}(2,2)$ 90\% relative width, $0.730$.  
PCGPwM results in the smallest intervals all but two of the model parameters. 
This constraining of plausible parameter region is attributed to the improved surrogate offering more precise predictions of  simulation responses.   
The benefit for this case study is that the analysis using the proposed surrogate method provides the uncertainty estimates for the parameter of the Fayans EDF model in contrast to previous studies.  
We find that there are considerable differences in the resulting interval widths, with some parameters (e.g.,  $f^\xi_\mathrm{ex}$ and $h^\xi_+$) being estimated more precisely,  and some parameters (e.g., $E/A, K,$ and $L$) having smaller precision improvements.
To more concretely understand the benefit, if we only use the complete output data to build our surrogate, the resulting posterior intervals would have been $6\%$--$58\%$ wider.

\begin{table}[ht]
    \centering
    \caption{Posterior 90\% widths relative to their respective ranges for the 13-dimensional parameter using different emulation techniques.}
    \label{tbl:posttheta_ci}
    \begin{tabular}{c|ccccc}
    \hline
    EDF Parameters              & PCGPwM & PCGP-kNN & colGP & Complete Data & Prior, $\mathrm{Beta}(2,2)$\\ \hline 
    $\rho_{\mathrm{eq}}$    & \textbf{0.355} & 0.500 & 0.546 & 0.491 & 0.730 \\ 
    $E/A$                   & \textbf{0.345} & 0.463 & 0.389 & 0.607 & 0.730 \\ 
    $K$                     & 0.455 & 0.493 & \textbf{0.361} & 0.643 & 0.730 \\
    $J$                     & \textbf{0.303} & 0.557 & 0.535 & 0.497 & 0.730 \\
    $L$                     & \textbf{0.437} & 0.484 & 0.393 & 0.576 & 0.730 \\
    $h^{\mathrm{v}}_{2-}$   & \textbf{0.370} & 0.462 & 0.401 & 0.587 & 0.730 \\
    $a^{\mathrm{s}}_{+}$    & \textbf{0.421} & 0.569 & 0.654 & 0.450 & 0.730 \\
    $h^{\mathrm{s}}_\nabla$ & \textbf{0.337} & 0.405 & 0.351 & 0.637 & 0.730 \\
    $\kappa$                & \textbf{0.339} & 0.479 & 0.461 & 0.614 & 0.730 \\
    $\kappa^\prime$         & \textbf{0.198} & 0.319 & 0.273 & 0.421 & 0.730 \\
    $f^\xi_\mathrm{ex}$     & \textbf{0.125} & 0.345 & 0.246 & 0.300 & 0.730 \\
    $h^\xi_\nabla$          & 0.386 & 0.530 & \textbf{0.347} & 0.536 & 0.730 \\
    $h^\xi_+$               & \textbf{0.128} & 0.367 & 0.215 & 0.254 & 0.730 \\
    \hline
    \end{tabular}
\end{table}

\section{Conclusion} \label{sec:conclusion}
This article details a new surrogate construction method developed to handle missing data.  The construction relies on an imputation of the missing data and a covariance adjustment to account for the added uncertainty due to imputation.  
This method is efficient and adds minimal burden on computations.  The surrogate construction is effective in ignoring entirely missing data in a multivariate output context.  Furthermore, it retains the efficiency of modern approaches to building surrogates of high-dimensional output data.

It is expected that as nascent models are used in large computing environments, partially observed output data will become more prevalent.
For example, \citet{lin2021uncertainty} conduct their inference by constructing a missingness classifier in addition to a surrogate using available simulation data.
When missingness is important to inference, the method proposed in this article could be used in to help with the surrogate to prevent tractability issues.  

\subsection*{Supplementary Materials}
The supplementary materials contain (i) proofs for Theorems \ref{thm:recovery} and \ref{thm:ignorefullmissingrows}, (ii) technical details for computation, (iii) descriptions of test functions, (iv) full results from the numerical experiments, (v) the original scaling of the Fayans EDF parameter space, and (vi) the code and data for the Fayans EDF case study.

\if0\blind
{
\subsection*{Acknowledgments}
We thank the editor, AE, two anonymous referees, and Earl Lawrence and Kelly Moran for their valuable feedback for improving this article's exposition.
We are grateful to Jared O'Neal and Paul-Gerhard Reinhard for developing the Fayans EDF model employed here. We gratefully acknowledge the computing resources provided on Bebop, a high-performance computing cluster operated by the Laboratory Computing Resource Center at Argonne National Laboratory.
This research was supported in part through the computational resources and staff contributions provided for the Quest high-performance computing facility at Northwestern University, which is jointly supported by the Office of the Provost, the Office for Research, and Northwestern University Information Technology.
} \fi

\subsection*{Disclosure statement}
The authors report there are no competing interests to declare.

\bibliographystyle{asa}
\bibliography{missingsurrogate}

\appendix

\input{supplement}

\end{document}

%% file: preamble.tex
\usepackage[utf8]{inputenc}
\usepackage{amsmath}
\usepackage{amssymb}
\usepackage{amsthm}
\usepackage{bm}
\usepackage{upgreek}
\usepackage{setspace}
\usepackage{xspace}
\usepackage{url}
\usepackage[group-separator={,},group-minimum-digits=4]{siunitx}

\usepackage[ruled]{algorithm2e}

\usepackage{natbib}

\usepackage{enumitem}
\usepackage{graphicx}
\usepackage{subcaption}
\usepackage{placeins}
\usepackage{booktabs}
\graphicspath{{./fig/}}

\usepackage{lineno}

\newcommand{\T}{\mathsf{T}}
\newcommand{\setX}{\mathcal{X}}

\newcommand{\obsdata}{F}
\newcommand{\projobsdata}{G}

\newcommand{\Cov}{\mathrm{Cov}}

\renewcommand{\hat}{\widehat}
\newcommand{\diag}{\mathrm{diag}}

\newcommand{\ourtitle}{Constructing a simulation surrogate with partially observed output}
\newcommand{\ourauthors}{Moses Y-H. Chan, Matthew Plumlee\thanks{
    This material is based upon work supported by NSF grants OAC 2004601, DMS 1953111, 1952897.}\\
    Department of Industrial Engineering and Management Sciences,\\ Northwestern University\\
    and \\
    Stefan M.~Wild\thanks{
    This material is based upon work supported by the U.S.\ Department of Energy, Office of Science, Office of Advanced Scientific Computing Research SciDAC and applied mathematics programs under contract DE-AC02-05CH11231, and by NSF grants OAC 2004601 and DMS 1952897.} \\
    Applied Mathematics and Computational Research Division, \\Lawrence Berkeley National Laboratory
    \\ NAISE, Northwestern University}
\theoremstyle{definition}
\newtheorem{theorem}{Theorem}

\newtheorem{lemma}{Lemma}

%% file: supplement.tex
\newpage
\spacingset{1.5}

{\LARGE \bf Supplementary materials to ``\ourtitle''}

\if0\blind
{
    \begin{center}
    Moses Y-H. Chan, Matthew Plumlee\\
    Department of Industrial Engineering and Management Sciences,\\ Northwestern University\\
    and \\
    Stefan M.~Wild \\
    Applied Mathematics and Computational Research Division, \\Lawrence Berkeley National Laboratory
    \\ NAISE, Northwestern University
    \end{center}
} \fi

This document includes the supplementary material to the main article ``Constructing a simulation surrogate with partially observed output''. 
The supplementary materials are organized in the following order:
    \begin{enumerate}[label=]
        \item Section A -- G: Proofs and technical details for computation.
        \item Section H: Test functions for the numerical experiments.
        \item Section I: Full numerical experiment results.
        \item Section J: Original scaling of the Fayans EDF parameter space.
        \item Section K: Code and data for the Fayans EDF case study.
    \end{enumerate}

\section{Proof of Theorem \ref{thm:recovery}} 
\begin{proof}
    First, we note that $\tilde{\bm\projobsdata}_{i,\cdot} = \bm\projobsdata_{i,\cdot}$ and $u_{ik} = 0$ for all $k$ because
    \[ \tilde{\bm\projobsdata}_{i,\cdot} = \bm\Phi^\mathsf{T}  \bm B_{\mathcal{J}(i),\cdot}^\T  \bm B_{\mathcal{J}(i),\mathcal{J}(i)}^{-1} \bm \obsdata_{i,\mathcal{J}(i)}^\mathsf{T} = \bm \Phi^\mathsf{T}  \bm B  \bm B^{-1} \bm \obsdata_{i,\cdot}^\mathsf{T}  = \bm \projobsdata_{i, \cdot}, \]
    and for all $k$,
    \[ u_{ik} = \bm \Phi_{\cdot, k}^\T \left(\bm B - \bm B \bm B^{-1} \bm B^\T \right) \bm \Phi_{\cdot, k} = 0. \]
    Then we observe that$\bm r_k(\bm \theta_i)$ is the $i$ row of $\bm R_k$, and thus $\bm r_k^\mathsf{T}(\bm \theta_i) \bm R_k^{-1} = \bm e_i$, where $\bm e_i$ is the $i$th vector of the identity matrix.  We conclude that
    $$\tilde{\mu}_k(\bm \theta_i) = \bm e_i^\mathsf{T} \tilde{\bm \projobsdata}_{\cdot,k} = \tilde{\bm \projobsdata}_{ik} = \bm \projobsdata_{ik} = \bm f(\bm \theta_i)^\T \bm \Phi_{\cdot,k},$$
    and
    $$\tilde{\sigma}_k^2(\bm \theta_i) = \lambda_k \left( \rho_k(\bm \theta_i,\bm \theta_i) - \bm e_i^\mathsf{T} \bm r_k(\bm \theta_i) \right) =\lambda_k \left( \rho_k(\bm \theta_i,\bm \theta_i) -\rho_k(\bm \theta_i,\bm \theta_i)\right) = 0.$$
\end{proof}

\section{Boundedness of $w_{ik}$}\label{appx:boundedness_wik}

\begin{lemma} \label{prop: w_ik_bound}
    For any $\mathcal{J}(i) \subseteq \{1,\ldots,m\}$,  $w_{ik}$ given in expression \eqref{eqn:gtildevar} is bounded between 0 and $1.$
\end{lemma} 
    
    \begin{proof}
    For this proof we will shorten $\mathcal{J}(i)$ to $\mathcal{J}$ and drop the subscript $i$ from $w_{ik}$. 
      We have that
      \[\lambda_k w_k  = \bm \Phi_{k}^\mathsf{T} \left(\bm B - \bm B_{\mathcal{J}, \cdot}^\T \bm B_{\mathcal{J}\mathcal{J}}^{-1} \bm B_{\mathcal{J}, \cdot} \right) \bm \Phi_{k}.\]
    
    Without loss of generality, say that $\mathcal{I} = \{1,\ldots,t \}$ and $\mathcal{J} = \{t+1,\ldots,m\}$.  
The matrix $\bm B$ can then be written as
    \[\bm B = \begin{pmatrix} \bm B_{\mathcal{I}\mathcal{I}} & \bm B_{\mathcal{I}\mathcal{J}} \\ \bm B_{\mathcal{J}\mathcal{I}} & \bm B_{\mathcal{J}\mathcal{J}} \end{pmatrix}.\]
    Letting $\bm v = \bm \Phi_{\mathcal{I}k}$ and $\bm u = - \bm B_{\mathcal{J}\mathcal{J}}^{-1} \bm B_{\mathcal{J}\mathcal{I}} \bm v$, we have that 
    \[\bm v^\T \bm B_{\mathcal{I}\mathcal{J}} \bm B_{\mathcal{J}\mathcal{J}}^{-1} \bm B_{\mathcal{J}\mathcal{I}} \bm v  = -\bm u^\mathsf{T} \bm B_{\mathcal{J}\mathcal{J}} \bm u - 2 \bm u^\mathsf{T}\bm B_{\mathcal{J}\mathcal{I}} \bm v, \]
    and conclude that 
    \begin{align}
        \bm v^\mathsf{T} &\left(\bm B_{\mathcal{I}\mathcal{I}} - \bm B_{\mathcal{I} \mathcal{J}} \bm B_{\mathcal{J}\mathcal{J}}^{-1} \bm B_{\mathcal{I} \mathcal{J}}^\T \right) \bm v =  \bm v^\T  \bm B_{\mathcal{I}\mathcal{I}} \bm v  + \bm u^\mathsf{T} \bm B_{\mathcal{J}\mathcal{J}} \bm u + 2 \bm u^\mathsf{T}\bm B_{\mathcal{J}\mathcal{I}} \bm v  \nonumber \\
         =& \begin{pmatrix} \bm v & \bm u \end{pmatrix} \begin{pmatrix} \bm B_{\mathcal{I}\mathcal{I}} & \bm B_{\mathcal{I}\mathcal{J}} \\ \bm B_{\mathcal{J}\mathcal{I}} & \bm B_{\mathcal{J}\mathcal{J}} \end{pmatrix} \begin{pmatrix} \bm v \\ \bm u \end{pmatrix}. \label{eq:pd_part}
    \end{align}
    Note that $\bm B = \bm \Phi ( \bm \Lambda - \varepsilon \bm I) \bm \Phi^\mathsf{T} + \varepsilon \bm I$. 
    By construction, since each diagonal element of $\bm \Lambda$ is larger than  $\varepsilon$, $\bm B$ is the sum of two positive-definite matrices and is thus positive definite. Since the right-hand side in (\ref{eq:pd_part}) is larger than zero, we have that $w_{k} \geq 0$.  
    
    Now, since $\bm B_{\mathcal{J}\mathcal{J}}$ is positive definite,  $\bm B_{\mathcal{J}\mathcal{J}}^{-1}$ is positive definite and thus 
    \[(\bm B_{\mathcal{J}, \cdot} \bm \Phi_{k})^\mathsf{T} \bm B_{\mathcal{J}\mathcal{J}}^{-1} \left(\bm B_{\mathcal{J}, \cdot} \bm \Phi_{k}\right) \geq \bm 0.\]
    We conclude that 
    $$ \bm \Phi_k^\mathsf{T} \left(\bm B - \bm B_{\mathcal{J},\cdot}^\T \bm B_{\mathcal{J}\mathcal{J}}^{-1}\bm B_{\mathcal{J},\cdot} \right) \bm \Phi_k \leq \bm \Phi_k^\mathsf{T} \bm B \bm \Phi_k.$$
    Finally, we note that by the orthogonality of $\bm \Phi$ and letting $\bm e_k$ be the $k$th  vector of the identity matrix, we have that
    \[\bm \Phi_k^\mathsf{T} \bm B \bm \Phi_k=  \bm \Phi_k^\mathsf{T} (\bm \Phi (\bm \Lambda - \varepsilon \bm I) \bm \Phi^\mathsf{T} + \varepsilon \bm I) \bm \Phi_k = \bm e_k^\T (\bm \Lambda - \varepsilon \bm I) \bm e_k + \varepsilon = \lambda_k,\]\
    and thus  $w_{k} \leq 1$. 
    \end{proof}

\begin{section}{Inverse of matrices with an increasing diagonal}
\label{appx:lemmalimitinv}
\begin{lemma} \label{lemma:limitinv}
    If the norm of an $n$-by-$n$ matrix $\bm A$ is bounded (i.e., $\lVert \bm A\rVert < L$), then $\lim_{\eta\rightarrow \infty} (\bm A + \eta \bm I)^{-1} = \bm 0.$ 
    \end{lemma}
    \begin{proof}
        For $\eta \geq L, \lVert \eta ^{-1}\bm A\rVert < 1$. Consequently, for the partial sum $\bm T_r = \bm I + \sum_{j=1}^r (-1)^j \eta^{-j} \bm A^j$, we have that  $\lim_{r\rightarrow\infty} \bm T_r = (\bm I + \eta^{-1} \bm A)^{-1}$ \citep[p.135]{gentle2007matrix}.  Consider the norm of the partial sum,
        \[\lVert \bm T_r \rVert = \left\lVert \bm I + \sum_{j=1}^r (-1)^j \eta^{-j} \bm A^j \right\rVert \leq \lVert \bm I \rVert + \sum_{j=1}^r \eta^{-j} \lVert \bm A \rVert^j \leq n  - 1+ \frac{1 - \eta^{-(r+1)}}{1-\eta^{-1}} \rightarrow n - 1+ \frac{1}{1-\eta^{-1}} , \] 
        as $r\rightarrow \infty.$ The sequence of matrices $\{\bm T_r\}_{r=1}^\infty$ uniformly converges over the entries,
        therefore,  \[\lim_{\eta\rightarrow \infty} (\bm A + \eta \bm I)^{-1} = \lim_{\eta\rightarrow \infty} \lim_{r\rightarrow \infty}  \eta^{-1} \bm T_r = \lim_{r\rightarrow \infty} \lim_{\eta\rightarrow \infty} \eta^{-1} \left(\bm I +\sum_{j=1}^r  (-1)^j \eta^{-j} \bm A^j\right) = \bm 0. \]
    \end{proof}
\end{section}

\section{Proof of Theorem \ref{thm:ignorefullmissingrows}} 
\newcommand{\Q}{\mathcal{Q}}
\renewcommand{\S}{\mathcal{S}}
\begin{proof}
     From Lemma \ref{prop: w_ik_bound} in Supplementary Material \ref{appx:boundedness_wik}, without loss of generality, partition $\{1, \ldots, n\}$ into $\mathcal{Q} = \{1, \ldots, t\}$ and $\mathcal{S} = \{t+1, \ldots, n\}$, where $w_{ik} < 1$ for $i \in \mathcal{Q}$ and $w_{ik} = 1$ for $ i \in \mathcal{S}.$  Let $\bm Q_k$ be an $n$-by-$n$ diagonal matrix with diagonal entries $\left(\frac{w_{1k}}{(1-w_{1k})^\alpha}, \ldots, \frac{w_{tk}}{(1-w_{tk})^\alpha}, 0 ,\ldots, 0 \right)$. Let $\bm S$ be $\begin{bmatrix} 0_{t\times(n-t)} \\  \bm I_{n-t} \end{bmatrix}$.  Then, once $\eta$ is big enough such that for all sets $J \subseteq \{1,\ldots,m\}$ with at least one entry, $\frac{w_{ik}}{(1-w_{ik})^\alpha} < \eta$, the adjusted covariance matrix can be rewritten as
    \[ \bm R_k + \beta_k \diag(\bm v_k) = {\bm R}_k + \beta_k (\bm Q_k + \eta \bm S\bm S^\T). \]
    Without loss of generality, let $\beta_k = 1$.   We drop the subscript $k$ for the remainder of the proof for brevity. 
    Partition the covariance matrix as 
     \[ \bm R + \bm Q + \eta \bm S\bm S^\T = \begin{bmatrix} (\bm R + \bm Q)_{\Q\Q} & \bm R_{\Q\S} \\ \bm R_{\S\Q} & \bm R_{\S\S} + \eta \bm I \end{bmatrix}. \]
    Recall that for an invertible matrix $\bm M = \begin{bmatrix}    \bm W & \bm X \\ \bm Y & \bm Z    \end{bmatrix}$, the inverse 
    \citep[p.95]{gentle2007matrix} is
    \[ \bm M^{-1} = \begin{bmatrix}    \bm W & \bm X \\ \bm Y & \bm Z    \end{bmatrix}^{-1} 
    = \begin{bmatrix}   (\bm W - \bm X\bm Z^{-1}\bm Y)^{-1} & - \bm W^{-1} \bm X (\bm Z - \bm Y\bm W^{-1}\bm X)^{-1} \\
                        - \bm Z^{-1} \bm Y (\bm W - \bm X\bm Z^{-1}\bm Y)^{-1} & (\bm Z - \bm Y \bm W^{-1}\bm X )^{-1}
    \end{bmatrix}.\]
    The norm $\lVert \bm R_{\S\S} \rVert$ is bounded by the largest singular value of $\bm R_{\S\S}$.  Using Lemma \ref{lemma:limitinv} in Supplementary Material \ref{appx:lemmalimitinv}, 
    \[\lim_{\eta \rightarrow \infty}(\bm R_{\S\S} +  \eta \bm I)^{-1} = \bm 0,\]
    \[\lim_{\eta\rightarrow\infty}\left(\bm R_{\S\S} + \eta \bm I - \bm R_{\S\Q}(\bm R+\bm Q)_{\Q\Q}^{-1} \bm R_{\Q\S}\right)^{-1} = \bm 0.\]  
    By the Woodbury identity \citep[p.221]{gentle2007matrix}, 
    \begin{align*}
        &\left((\bm R+\bm Q)_{\Q\Q} - \bm R_{\Q\S} (\bm R_{\S\S} + \eta \bm I)^{-1} \bm R_{\S\Q}\right)^{-1} \\      
        &= (\bm R+\bm Q)_{\Q\Q}^{-1} + (\bm R+\bm Q)_{\Q\Q}^{-1} \bm R_{\Q\S} \left( \bm R_{\S\S} + \eta \bm I - \bm R_{\S\Q}(\bm R+\bm Q)_{\Q\Q}^{-1} \bm R_{\Q\S}\right)^{-1} \bm R_{\S\Q} (\bm R+\bm Q)_{\Q\Q}^{-1} \\
        &\rightarrow (\bm R+\bm Q)_{\Q\Q}^{-1}.
    \end{align*}
    To summarize, as $\eta \rightarrow \infty$,
    \[\left(\bm R + \bm Q + \eta \bm S \bm S^\T\right)^{-1} \rightarrow \begin{bmatrix}
    (\bm R+\bm Q)_{\Q\Q}^{-1} & \bm 0_{t\times(n-t)} \\
    \bm 0_{(n-t)\times t} & \bm 0_{(n-t)\times(n-t)} 
    \end{bmatrix}.\]
    Therefore, for any length-$n$ vector $\bm a$, 
    \[\left(\left(\bm R + \bm Q + \eta \bm S \bm S^\T\right)^{-1} \bm a\right)_\Q \rightarrow \left(\bm R+\bm Q\right)^{-1}_{\Q\Q} \bm a_{\Q},\]
    \[\left(\left(\bm R + \bm Q + \eta \bm S \bm S^\T\right)^{-1} \bm a\right)_\S \rightarrow \bm 0.\]
    
    Pulling these results through our equations for predictions (i.e., \eqref{eqn:tildemu_k} and \eqref{eqn:tildesigma_k}), we conclude that the limits of $\tilde{\mu}_k(\bm \theta)$ and $\tilde{\sigma}^2_k(\bm \theta)$ will not depend on any row $i$ with $w_{ik}=1$, and thus the result holds.
\end{proof}

\section{Estimating $\bm \Phi$} 
The principal component matrix $\bm \Phi$ is estimated with missing values in the observed simulation output $\bm F$.  
For a certain row $i$, the index set $\mathcal{J}(i)$ contains the indices where data are not missing.  Denote the missing index set as ${\sim}\mathcal{J}(i) = \{1, \ldots, m \} \setminus \mathcal{J}(i).$  
The principal component matrix $\bm \Phi$ is approximated via Algorithm \ref{alg:emPhi} which mirrors the typical EM structure.  The simulation output $\bm F$ is assumed to be scaled to zero mean and unit variance in its columns.
The E step is standard, employing conditional equations regarding multivariate normal distribution.  However, we note that the M step does not result in the exact optimizer since $\epsilon_\text{M} > 0$.  In our implementation, $\epsilon_\text{M} = 10^{-5}$ provides reasonable performance and ensures the procedure is numerically stable.

\SetKwComment{Comment}{/* }{ */}
\begin{algorithm}
\caption{EM algorithm for estimating $\bm \Phi$}\label{alg:emPhi}
\DontPrintSemicolon
\LinesNumbered
\Comment{Initialization with column means of non-missing values}
\For{$j \in \{1, \ldots, m\}$}{
    Let $\mathcal{I}(j) = \{i: j\in\mathcal{J}(i)\}.$\;
    \lForAll(\Comment*[f]{missing}){$i \notin \mathcal{I}(j)$}{
        $\tilde{\bm f}_j(\bm \theta_i) = \frac{1}{|\mathcal{I}(j)|}\sum_{i\in\mathcal{I}(j)} \bm f_j(\bm \theta_i).$
        }
    \lForAll(\Comment*[f]{non-missing}){$i \in \mathcal{I}(j)$}{
        $\tilde{\bm f}_j(\bm \theta_i) = \bm f_j(\bm \theta_i).$
        }
    }
\Comment{EM algorithm}
\While{convergence criterion not met}{
    Obtain $\bm\Phi, \bm\Lambda$ via SVD of $\tilde{\bm  F} = (\tilde{\bm f}(\bm \theta_1)^\T, \ldots, \tilde{\bm f}(\bm \theta_n)^\T)^\T$. \Comment*[r]{M step} 
    \For(\Comment*[f]{E step}){$i \in \{1, \ldots, n\}$}{
            $\tilde{\bm f}_{{\sim}\mathcal{J}(i)}(\bm \theta_i) = \bm \Phi_{{\sim}\mathcal{J}(i), \cdot}(\Lambda - \varepsilon \bm I) \bm \Phi_{\mathcal{J}(i), \cdot}^\mathsf{T} \left(\bm \Phi_{\mathcal{J}(i), \cdot}(\Lambda - \varepsilon \bm I) \bm \Phi_{\mathcal{J}(i), \cdot}^\mathsf{T}  + \epsilon_\text{M} \bm I\right)^{-1} \bm f_{\mathcal{J}(i)}(\bm \theta_i).$\;
    }
}
\end{algorithm}

\FloatBarrier

\section{Computing Inverse of $\bm B_{\mathcal{J}(i), \mathcal{J}(i)}$} 
The main inferences introduced in \eqref{eqn:gtildemean} and \eqref{eqn:gtildevar} depend on inverting the matrix $\bm B_{\mathcal{J}(i),\mathcal{J}(i)}.$  Recall the definition of the covariance matrix 
\[\bm B = \bm \Phi \, \diag\left(\lambda_1-\varepsilon, \ldots, \lambda_\kappa - \varepsilon\right) \, \bm \Phi^\T + \varepsilon \bm I. \]
Let $\bm \Lambda_\varepsilon = \diag\left(\lambda_1-\varepsilon, \ldots, \lambda_\kappa - \varepsilon\right)$. Now consider the inverse of $\bm B$ by using the Woodbury matrix identity,
\begin{align*}
    \bm B^{-1} &= \left(\varepsilon \bm I + \bm \Phi \bm \Lambda_\varepsilon  \bm \Phi^\T  \right)^{-1} \\
    &= \varepsilon^{-1} \bm I - \varepsilon^{-2} \bm \Phi \left(\bm \Lambda_\varepsilon^{-1}  + \varepsilon^{-1}\bm I\right)^{-1} \bm \Phi^\T \\
    &= \varepsilon^{-1} \bm I - \varepsilon^{-1} \bm \Phi \,  \diag\left(\frac{\lambda_1-\varepsilon}{\lambda_1},\ldots, \frac{\lambda_\kappa - \varepsilon}{\lambda_\kappa}\right) \bm \Phi^\T \\
    &= \varepsilon^{-1} \left(\bm I - \bm \Phi \,  \diag\left(\frac{\lambda_1-\varepsilon}{\lambda_1},\ldots, \frac{\lambda_\kappa - \varepsilon}{\lambda_\kappa}\right) \bm \Phi^\T\right).
\end{align*}
The expression above is simple to compute and requires no matrix inversion.  But we are concerned about the inverse of the submatrix $\bm B_{\mathcal{J}(i), \mathcal{J}(i)}$, where
\begin{align*}
    \bm B_{\mathcal{J}(i),\mathcal{J}(i)} = \varepsilon \bm I + \bm \Phi_{\mathcal{J}(i),\cdot} \,\bm  \Lambda_\varepsilon \, \bm \Phi_{\mathcal{J}(i),\cdot}^\T.
\end{align*}
Also by Woodbury identity, we then have
\begin{align*}
    \bm B_{\mathcal{J}(i),\mathcal{J}(i)}^{-1} &= \left(\varepsilon \bm I + \bm \Phi_{\mathcal{J}(i),\cdot} \, \bm \Lambda_\varepsilon \,  \bm \Phi_{\mathcal{J}(i),\cdot}^\T\right)^{-1} \\
    &= \varepsilon^{-1} \bm I - \varepsilon^{-2}\bm  \Phi \left(\bm \Lambda_\varepsilon^{-1}  + \varepsilon^{-1}\bm \Phi_{\mathcal{J}(i),\cdot}^\T \, \bm \Phi_{\mathcal{J}(i), \cdot}\right)^{-1}\bm  \Phi^\T.
\end{align*}
The resulting expression cannot be further simplified since $\bm \Phi^\T_{\mathcal{J}(i),\cdot} \bm \Phi_{\mathcal{J}(i),\cdot}$ no longer equals the identity matrix or some simple known form.  However, the algebraic manipulation may still be beneficial when we consider only the inverse of the inner $\kappa$-by-$\kappa$ matrix, as opposed to the generally larger matrix $\bm B_{\mathcal{J}(i), \mathcal{J}(i)}.$

\section{Investigation on $\alpha$, $\beta_k$ values}
We construct surrogates for the four test functions with MCAR responses at 5\%.  The details of the test functions are given in Supplementary Material \ref{appx:testfuncs}.  
The root mean squared error (RMSE) between the function values and the surrogate's predicted values was evaluated with a set of holdout simulation runs.  Figure~\ref{fig:alpha_beta} shows the RMSE for selected $\alpha$ values with $\beta_k=1$ or $\beta_k$ optimized.  
The RMSE generally increases as $\alpha$ increases.  When $\beta_k$'s are fixed to be 1, the error increases for larger ranges of $\alpha$, specifically for the wingweight function, the error continuously increases for $\alpha > 0.$    
The benefit in including $\beta_k$'s is shown in the generally lower RMSE achieved.  We suggest to include $\beta_k$'s in the hyperparameter estimation in constructing the surrogate.  When $\beta_k$'s are optimized, smaller $\alpha$ values are preferred.  

\begin{figure}[ht]
    \centering
    \includegraphics[width=0.8\linewidth]{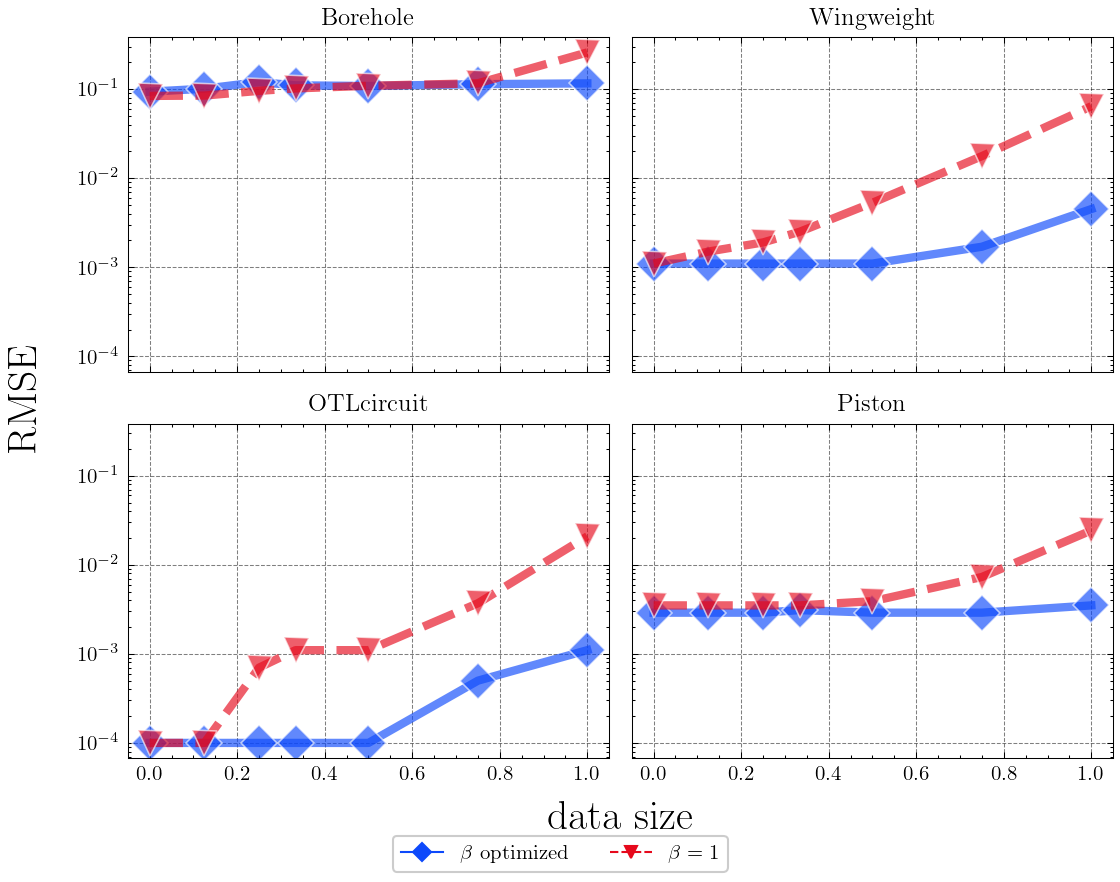}
    \caption{RMSE with $\alpha \in [0, 1]$ and $\beta_k$'s either optimized or $=1$.}
    \label{fig:alpha_beta}
\end{figure}

\FloatBarrier

\section{Test functions} \label{appx:testfuncs}
This section provides details of the four test functions used in evaluating the performance of the algorithm.  
The subscripts in the expressions in this section refer to the elements in each of the dimension of vector variables $\bm \theta$ and $\bm x$, whereas in all other sections, the subscripted $\bm\theta_i$ and $\bm x_j$ refer to the $i$th parameter and the $j$th locations, respectively.  The reference parameter $\bm\theta^\star$ refers to the center of the parameter space, for example $(0.5, 0.5, 0.5)^\T$ in a 3D unit hypercube.

For all functions, under MCAR, missingness is assigned with equal probability everywhere with the desired percentage.  Under MAR, missingness is randomly assigned to a subset of locations, according to the missing-at-random scheme reported in \citet{muzellec2020missing}.  For MNAR, the missingness mechanism is included in each function below.

\subsection{Modified Borehole function}
This test function is developed for the purpose of calibration and is modified from the Borehole function \citep{santner2018design}.  The function is defined as
\begin{equation} \label{eqn:borehole}
f(\bm \theta, \bm x) = \frac{2\pi (\theta_1 - x_1)}{2(\theta_2 / x_2^2) + \theta_3} \exp\left(\theta_4 x_2\right),
\end{equation}
where the ranges of the variables are included in Table \ref{tbl:varBorehole}.  The first four variables are to be tuned as the parameter, and the last two variables are fixed as locations.

\begin{table}[ht]
    \centering
    \begin{tabular}{ccc}
        \hline
        $\bm \theta$ & $\theta_1$       & [990, 1110]  \\
        & $\theta_2$       & [0.074, 1.12] \\
        & $\theta_3$       & [0.05, 0.5]    \\
        & $\theta_4$       & [-0.5, 0.5] \\
        $\bm x$     & $x_1$           & [700, 820] \\
        & $x_2$           & [0.05, 0.5] \\
        \hline
    \end{tabular}
    \caption{Variable ranges for borehole function in \eqref{eqn:borehole}.}
    \label{tbl:varBorehole}
\end{table}

Furthermore, a missingness of MNAR is introduced into the model evaluations, in the following manner:
\begin{equation*}
    f(\bm \theta, \bm x) \gets 
    \begin{cases}
        f(\bm \theta, \bm x), & \text{if $f(\bm \theta, \bm x) \leq c_1 f(\bm \theta^\star, \bm x)$} \\
        \mathrm{NA}, & \text{otherwise,}
    \end{cases}
\end{equation*}
where $\bm \theta^\star$ is the center of the standardized parameter space $U[0, 1]^4$. The constant $c_1>0$ is adjusted to achieve a desired probability of missingness.  
 
\subsection{Piston function}
The piston function \citep{ben2007modeling} is given as 
\begin{align}
\label{eq:piston}
    f(\bm \theta, \bm x) &= 2\pi\sqrt{\frac{M}{k+ S^2 \frac{P_0 V_0}{T_0} \frac{T_a}{V^2}}} \, , 
    \end{align}
    where
    \begin{align*}
    V &= \frac{S}{2k}\left(\sqrt{A^2 + 4k \frac{P_0 V_0}{T_0} T_a} - A\right), \\
    A&= P_0S + 19.62 M -\frac{kV_0}{S},
\end{align*}
and where the variable partition $(\bm \theta, \bm x)$ and their ranges are shown in Table~\ref{tbl:varPiston}. 

\begin{table}[ht]
    \centering
    \begin{tabular}{ccc}
        \hline
                 &  Variable        & Range \\ \hline
        $\bm \theta$ &  $k$             & [1000, 5000]  \\
                 &  $P_0$           & [90000, 110000] \\
                 &  $T_a$           & [290, 296]    \\ \hline
        $\bm x$      &  $M$             & [30, 60] \\
                 &  $S$             & [0.005, 0.02] \\
                 &  $V_0$           & [0.0002, 0.01] \\
                 &  $T_0$           & [340, 360] \\
        \hline
    \end{tabular}
    \caption{Variable ranges for the piston function in \eqref{eq:piston}.}
    \label{tbl:varPiston}
\end{table}

The function produces a MNAR that follows the ``logistic missing not-at-random'' scheme used in \citet{muzellec2020missing}, where the probability of an entry missing is according to a logistic model over the columns.   

\subsection{Wingweight function}
The wingweight function \citep[ch. 1]{forrester2008engineering} is given as 
\begin{align}
\label{eq:ww}
    f(\bm \theta, \bm x) &= 0.036 S_w^{0.758} W_{fw}^{0.0035} \left(\frac{A}{\cos^2(\Lambda)}\right)^{0.6} q^{0.006} \lambda^{0.04} \left(\frac{100 t_c}{\cos(\Lambda)}\right)^{-0.3} (N_z W_{dg})^{0.49} + S_w W_p,
\end{align}
and the variable partition $(\bm \theta, \bm x)$ and the respective ranges are shown in Table~\ref{tbl:varWingweight}.

\begin{table}[ht]
    \centering
    \begin{tabular}{ccc}
        \hline
                 &  Variable        & Range \\ \hline
        $\bm \theta$ &  $ A $           & [6, 10]  \\
                 &  $ \Lambda $     & [$-10^{\circ}, 10^{\circ}$] \\
                 &  $ q $           & [16, 45]    \\ 
                 &  $ \lambda $     & [0.5, 1]    \\ \hline
        $\bm x$      &  $ S_w $         & [150, 200] \\
                 &  $ W_{fw} $        & [220, 300] \\
                 &  $ t_c $         & [0.08, 0.18] \\
                 &  $ N_z $         & [2.5, 6] \\
                 &  $ W_{dg} $      & [1700, 2500] \\
                 &  $ W_p $         & [0.025, 0.08] \\
        \hline
    \end{tabular}
    \caption{Variable ranges for wingweight function in \eqref{eq:ww}.}
    \label{tbl:varWingweight}
\end{table}

The function produces a MNAR in the following manner:
\begin{equation*}
    f(\bm \theta, \bm x) = 
    \begin{cases}
        f(\bm \theta, \bm x), & \text{if $f(\bm \theta, \bm x) \leq c_2 f(\bm \theta^\star, \bm x)$} \\
        \mathrm{NA}, & \text{otherwise,}
    \end{cases}
\end{equation*}
where $\bm \theta^\star$ is the center of the standardized parameter space $U[0, 1]^4$. The constant $c_2>0$ is adjusted to achieve a desired probability of missingness.  
 
\subsection{OTLcircuit function}
The OTLcircuit function \citep{ben2007modeling} is given as
\begin{align}
\label{eq:otl}
    f(\bm \theta, \bm x) &= \frac{(V_{b1} + 0.74) \beta (R_{c2} + 9)}{\beta (R_{c2} + 9) + R_f} + \frac{11.35 R_f}{\beta (R_{c2} + 9) + R_f} + \frac{0.74 R_f \beta (R_{c2} + 9)}{(\beta (R_{c2} + 9) + R_f) R_{c1}}, 
    \end{align}
where
    \begin{align*}
    V_{b1} &= \frac{12 R_{b2}}{R_{b1} + R_{b2}},
\end{align*}
and the variable partition $(\bm \theta, \bm x)$ and the respective ranges are shown in Table~\ref{tbl:varOTLcircuit}.

\begin{table}[ht]
    \centering
    \begin{tabular}{ccc}
        \hline
                 &  Variable        & Range \\ \hline
        $\bm \theta$ &  $ R_f $           & [0.5, 3]  \\
                 &  $ \beta $         & [50, 300] \\ \hline
        $\bm x$      &  $ R_{b1} $        & [50, 150] \\
                 &  $ R_{b2} $        & [25, 70] \\
                 &  $ R_{c1} $        & [1.2, 2.5] \\
                 &  $ R_{c2} $        & [0.25, 1.2] \\
        \hline
    \end{tabular}
    \caption{Variable ranges for the OTLcircuit function in \eqref{eq:otl}.}
    \label{tbl:varOTLcircuit}
\end{table}

The function produces a MNAR that follows the``logistic missing not-at-random'' scheme used in \citet{muzellec2020missing}, where the probability of missing is according to a logistic model over the columns.

\section{Full results for surrogate comparison experiment}
This section presents the full simulation results following the numerical experiments detailed in the main article.  The surrogate methods compared are the proposed method PCGPwM, two simplistic imputation methods PCGP-kNN and PCGP-BR, EMGP, colGP, and GP-OM.
The surrogate methods are tested under nine missingness scenarios: (MCAR, 1\%), (MCAR, 5\%), (MCAR, 25\%), (MNAR, 1\%), (MNAR, 5\%), (MNAR, 25\%), (MAR, 1\%), (MAR, 5\%), and (MAR, 25\%), where results for the scenario (MNAR, 5\%) is reported in the main text of the article.  
To summarize the metrics of surrogate quality, the test RMSE is recorded to measure the accuracy of the surrogate; the 90\% coverage and the 90\% width are used to assess the uncertainty quantification of the surrogate.

\subsection{Results under MNAR}
The results under MNAR are presented, in the order of 1\%, 5\%, and 25\%.  

Figure \ref{fig:rmses1structured} shows the test RMSEs for the surrogate methods under 1\% missingness.  The simplistic imputation methods fail to reduce the errors (i.e., fail to improve surrogate predictions) as $N$ grows for two of the functions (borehole and wingweight).  The corresponding missingness mechanism is MNAR, depending on the unobserved function value. The simplistic imputation methods reduce the errors in the other two functions (piston and OTLcircuit). The RMSEs from PCGPwM and the other methods decrease as $N$ increases in all four functions.  The EMGP method consistently results in larger error than PCGPwM, and colGP.  The colGP method performs well across functions.  However, colGP will not be able to compete in terms of its computation time at higher output dimensions.

Figure~\ref{fig:coverages1structured} shows the 90\% coverage under 1\% missingness, and Figure~\ref{fig:avgintwidth1structured} shows the 90\% width.  GP-OM results in a high coverage in the borehole function but poor coverages in the other functions.  The 90\% widths show that GP-OM results in a confidence interval too wide in the borehole and too narrow in the others.  The EMGP method achieves the prescribed coverage in the piston and OTLcircuit functions, while undercovering in the borehole function.  The coverage from EMGP improves as $N$ grows and achieves the prescribed level for larger $N$s. PCGPwM, the simplistic imputation methods, and colGP result in adequate coverages in all functions.  However, the simplistic imputation methods maintain wider intervals than PCGPwM in all four cases.

\begin{figure}[ht]
    \centering
    \includegraphics[width=0.8\textwidth]{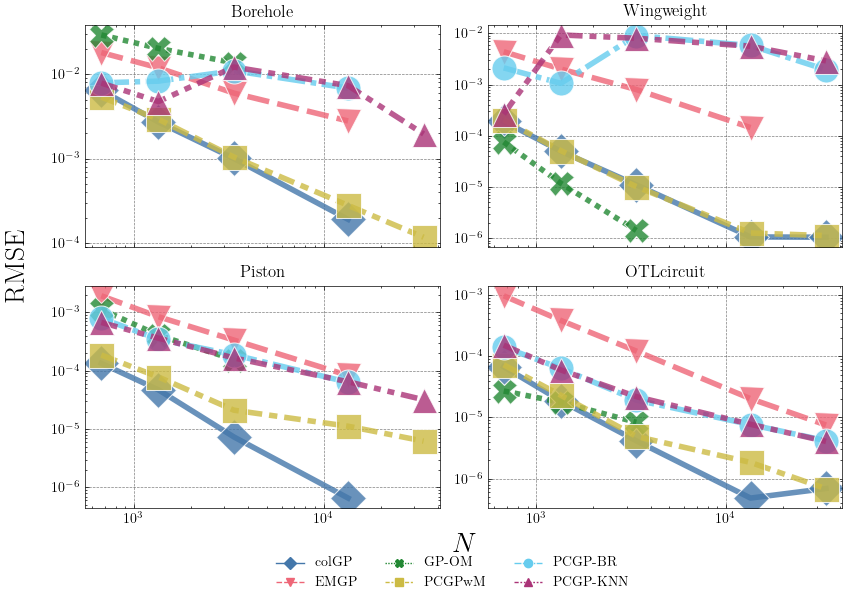}
    \caption{Comparison of prediction accuracy of surrogate methods, under 1\% MNAR.}
    \label{fig:rmses1structured}
\end{figure}

\begin{figure}[ht]
    \centering
    \includegraphics[width=0.8\textwidth]{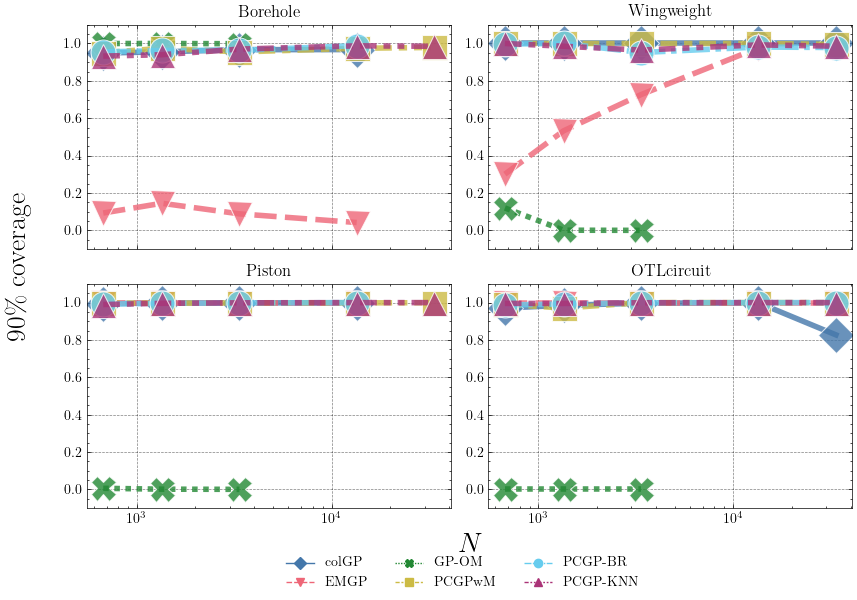}
    \caption{Comparison of 90\% coverage of surrogate methods, under 1\% MNAR.}
    \label{fig:coverages1structured}
\end{figure}

\begin{figure}[ht]
    \centering
    \includegraphics[width=0.8\textwidth]{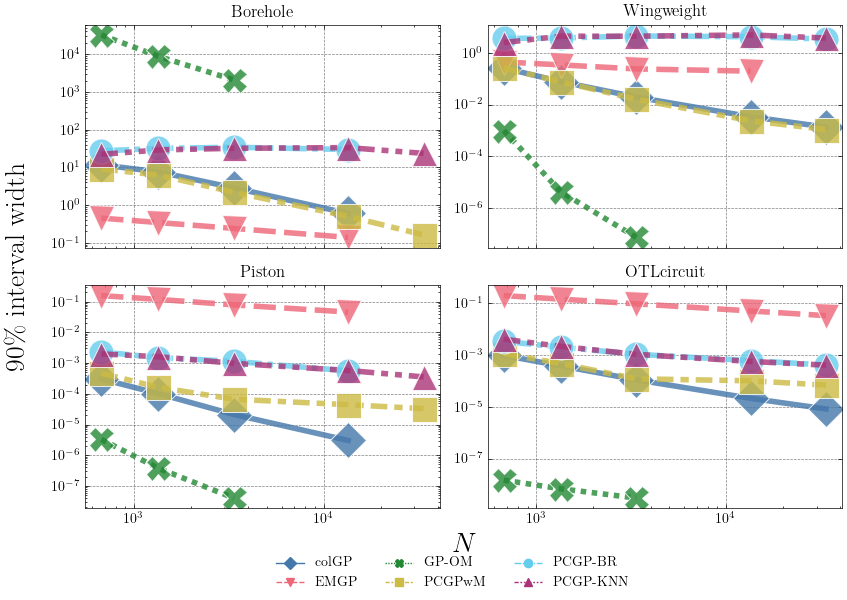}
    \caption{Comparison of 90\% width of surrogate methods, under 1\% MNAR.}
    \label{fig:avgintwidth1structured}
\end{figure}

Figure \ref{fig:rmses25structured} shows the test RMSEs under 25\% MNAR.  Similar conclusions to the 1\% case can be drawn about the simplistic imputation methods, where they fail to improve in two functions given more data.  The other methods exhibit similar behaviors in predictive accuracy.

Figure \ref{fig:coverages25structured} shows the coverages under 25\% MNAR and Figure~\ref{fig:avgintwidth25structured} shows the widths of the intervals.  We observe that GP-OM and the simplistic imputation methods result in similar behavior as the other cases, meaning that GP-OM results in high coverage in only one function, and the simplistic imputation method achieves high coverage with generally wide intervals.  PCGPwM results in adequate coverages with narrower intervals for all functions except the OTLcircuit function for $n=2500$. 

\begin{figure}[ht]
    \centering
    \includegraphics[width=0.8\textwidth]{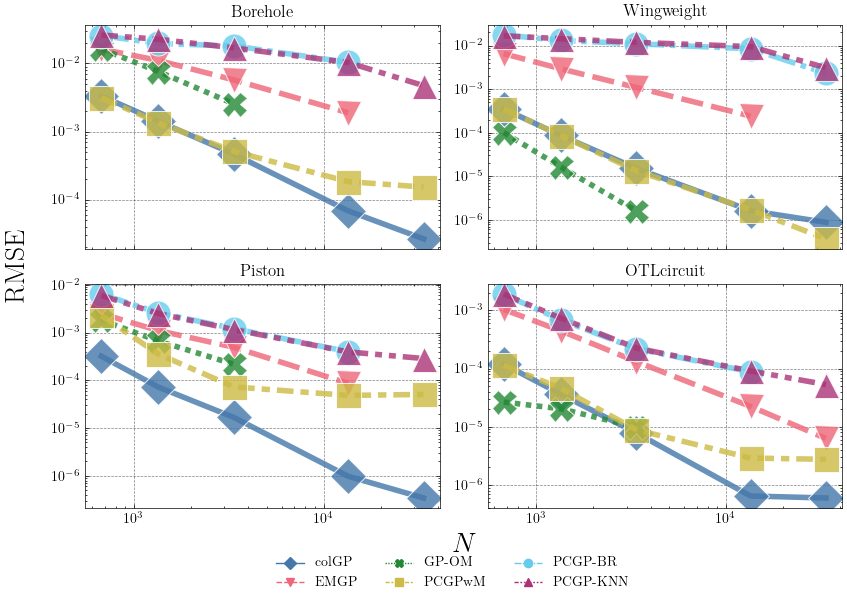}
    \caption{Comparison of prediction accuracy of surrogate methods, under 25\% MNAR.}
    \label{fig:rmses25structured}
\end{figure}

\begin{figure}[ht]
    \centering
    \includegraphics[width=0.8\textwidth]{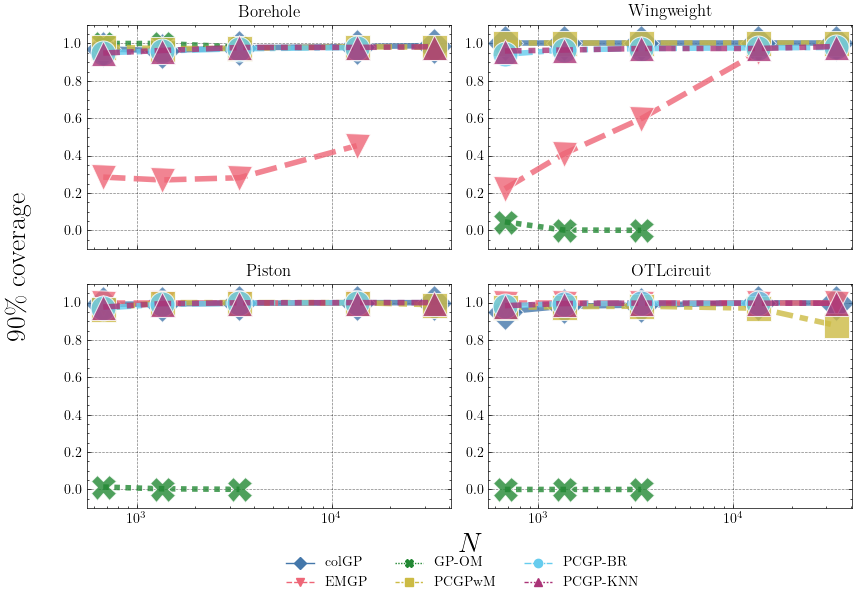}
    \caption{Comparison of 90\% coverage of surrogate methods, under 25\% MNAR.}
    \label{fig:coverages25structured}
\end{figure}

\begin{figure}[ht]
    \centering
    \includegraphics[width=0.8\textwidth]{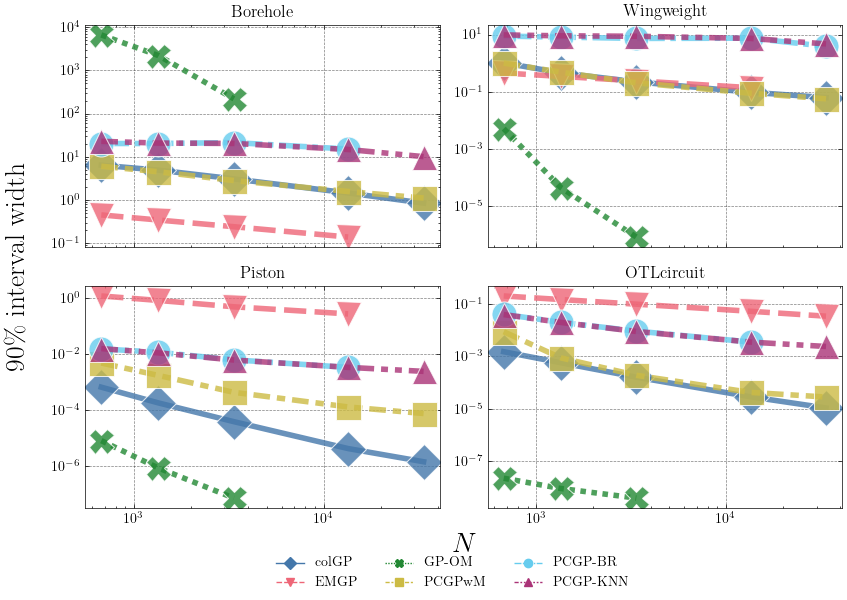}
    \caption{Comparison of 90\% width of surrogate methods, under 25\% MNAR.}
    \label{fig:avgintwidth25structured}
\end{figure}

\subsection{Results under MCAR}
The results under MCAR are presented, in the order of 1\%, 5\%, and 25\%.  

Figure \ref{fig:rmses1random} shows the test RMSEs for the surrogate methods under 1\% MCAR.  The RMSEs decrease as $N$ increases for all methods in all four test functions.  Among the principal component-based methods, PCGPwM generally achieves lower errors than the simplistic imputation methods.  The colGP method performs well across the functions, and the EMGP method generally results in higher error compared to the other methods.

Figure \ref{fig:coverages1random} shows the 90\% coverages, and Figure~\ref{fig:avgintwidth1random} shows the corresponding 90\% width of the intervals.  Similar to the MNAR case, GP-OM produces a high coverage in the borehole function, but near-zero coverages in the others.  The corresponding 90\% widths reflect the high coverage with wide intervals and the undercoverages with intervals that are too narrow.  The EMGP method achieves adequate coverage for the piston and OTLcircuit functions, but undercovers in the borehole function.  The coverage for the wingweight function improves as $N$ increases and achieves the prescribed level at larger $N$s.  All the other methods achieve adequate coverages, where the corresponding widths decrease as $N$ increases.  PCGPwM generally produces narrower widths than do the simplistic imputation methods while attaining comparable coverages.

\begin{figure}[ht]
    \centering
    \includegraphics[width=0.8\textwidth]{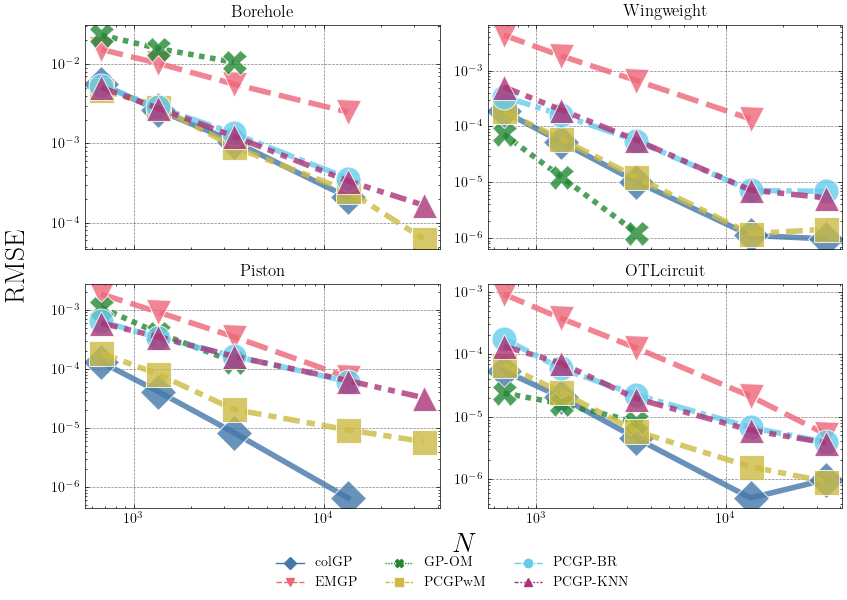}
    \caption{Comparison of prediction accuracy of surrogate methods, under 1\% MCAR.}
    \label{fig:rmses1random}
\end{figure}

\begin{figure}[ht]
    \centering
    \includegraphics[width=0.8\textwidth]{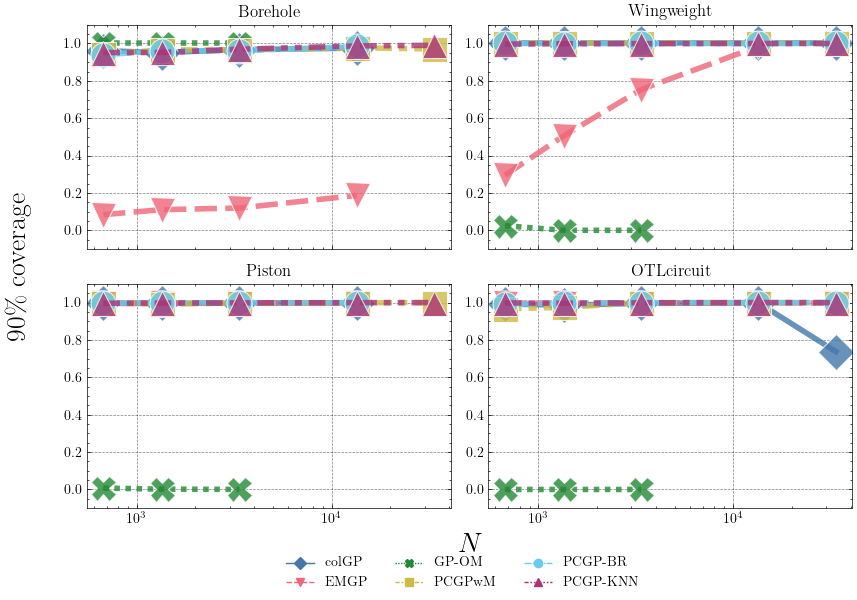}
    \caption{Comparison of 90\% coverage of surrogate methods, under 1\% MCAR.}
    \label{fig:coverages1random}
\end{figure}

\begin{figure}[ht]
    \centering
    \includegraphics[width=0.8\textwidth]{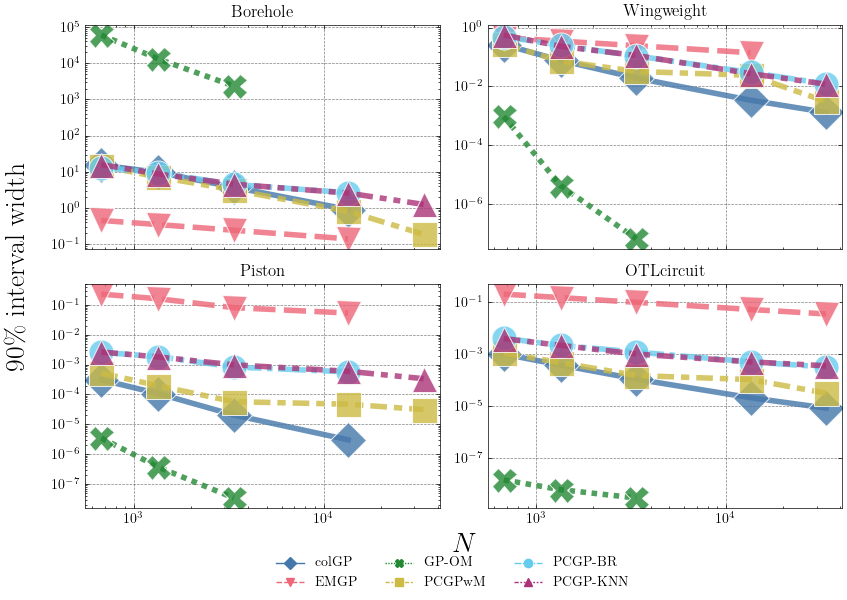}
    \caption{Comparison of 90\% width of surrogate methods, under 1\% MCAR.}
    \label{fig:avgintwidth1random}
\end{figure}

Figure \ref{fig:rmses5random} shows the test RMSEs under 5\% MCAR.  
Figures~\ref{fig:coverages5random}~and~\ref{fig:avgintwidth5random} show the 90\% coverages and the 90\% widths, respectively.  The results of the surrogate accuracy with  1\% MCAR extends to the scenario with 5\% MCAR, except the levels of coverage for PCGPwM has decreased in large $N$.  The 90\% widths decrease for all methods as $N$ increases.

\begin{figure}[ht]
    \centering
    \includegraphics[width=0.8\textwidth]{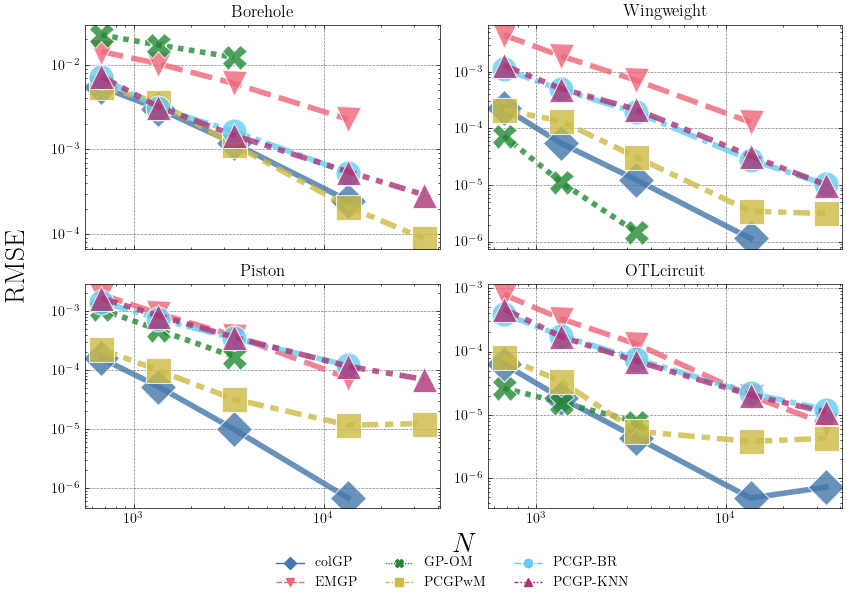}
    \caption{Comparison of prediction accuracy of surrogate methods, under 5\% MCAR.}
    \label{fig:rmses5random}
\end{figure}

\begin{figure}[ht]
    \centering
    \includegraphics[width=0.8\textwidth]{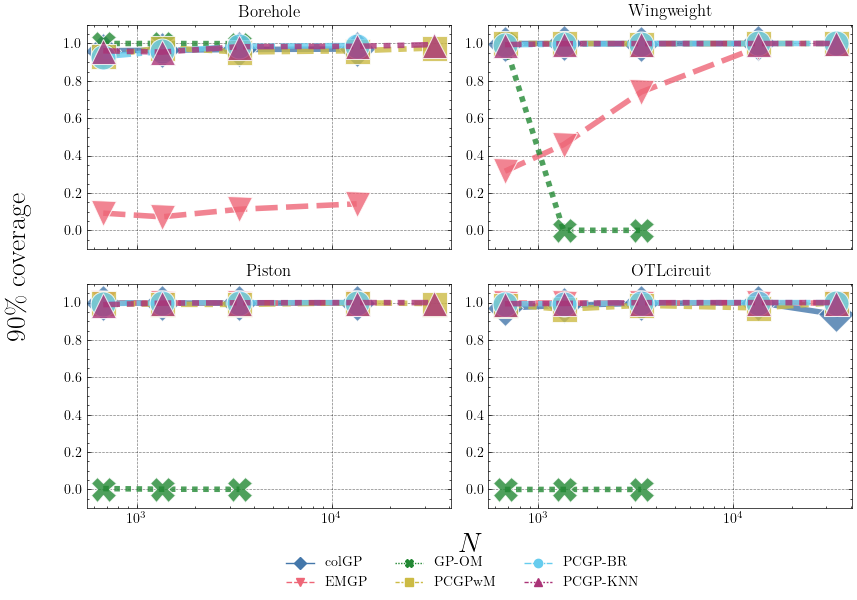}
    \caption{Comparison of 90\% coverage of surrogate methods, under 5\% MCAR.}
    \label{fig:coverages5random}
\end{figure}

\begin{figure}[ht]
    \centering
    \includegraphics[width=0.8\textwidth]{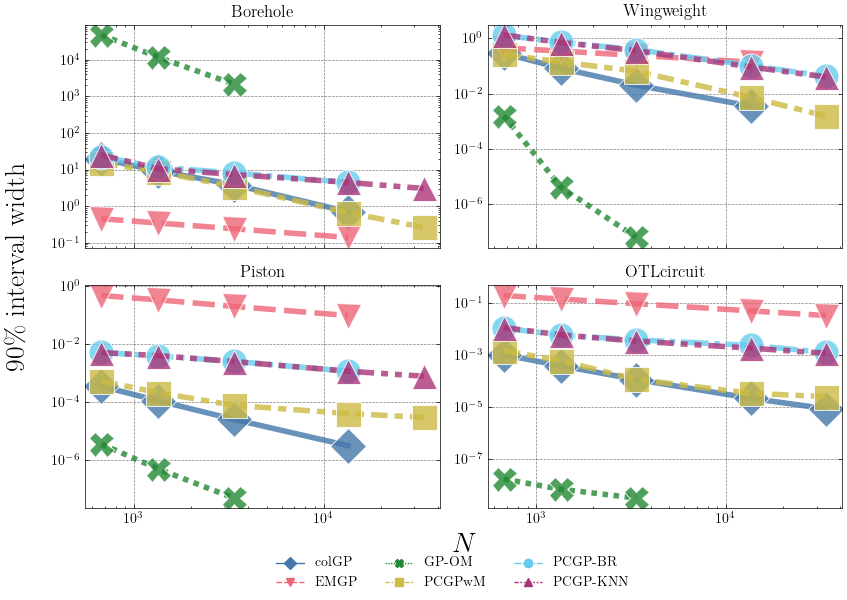}
    \caption{Comparison of 90\% width of surrogate methods, under 5\% MCAR.}
    \label{fig:avgintwidth5random}
\end{figure}

Figure \ref{fig:rmses25random} shows the test RMSEs under 25\% MCAR.  The RMSEs for all the methods generally decrease with $N$.
Figures~\ref{fig:coverages25random}~and~\ref{fig:avgintwidth25random} show the 90\% coverages and the 90\% widths, respectively.  Similar conclusions can be drawn for GP-OM, compared to lower missingness scenarios.  The simplistic imputation methods maintain adequate coverages while maintaining slightly wider intervals when compared with PCGPwM.  The EMGP method exhibits similar behavior, where coverage improves with the borehold and wingweight functions as $N$ increases.  

\begin{figure}[ht]
    \centering
    \includegraphics[width=0.8\textwidth]{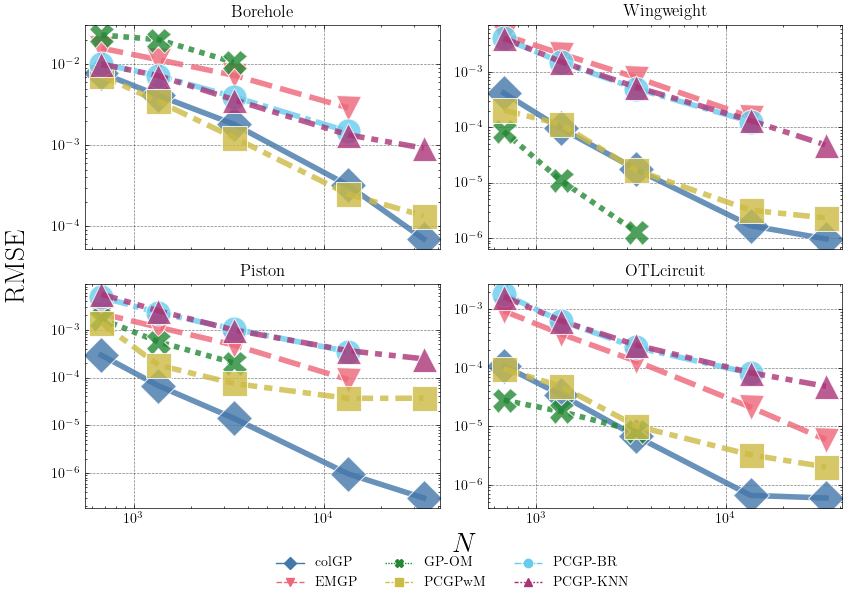}
    \caption{Comparison of prediction accuracy of surrogate methods, under 25\% MCAR.}
    \label{fig:rmses25random}
\end{figure}

\begin{figure}[ht]
    \centering
    \includegraphics[width=0.8\textwidth]{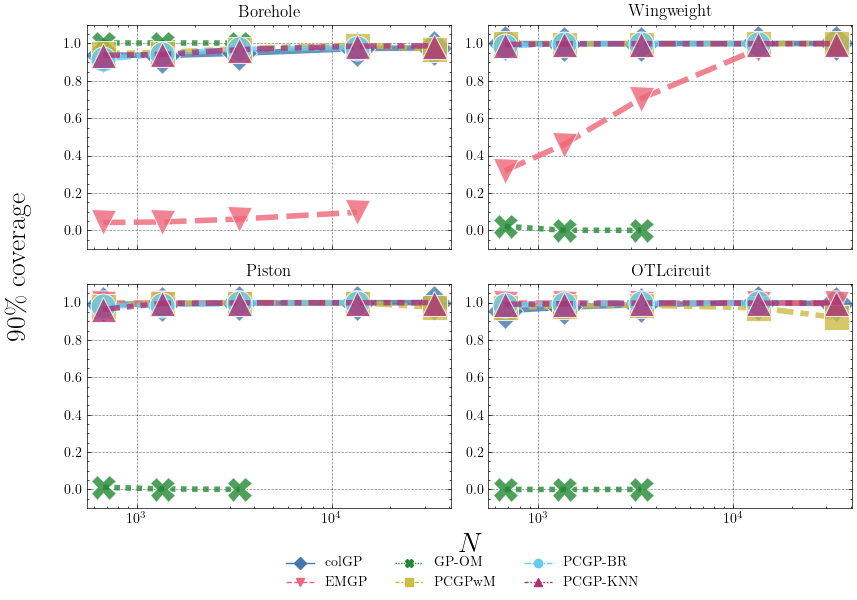}
    \caption{Comparison of 90\% coverage of surrogate methods, under 25\% MCAR.}
    \label{fig:coverages25random}
\end{figure}

\begin{figure}[ht]
    \centering
    \includegraphics[width=0.8\textwidth]{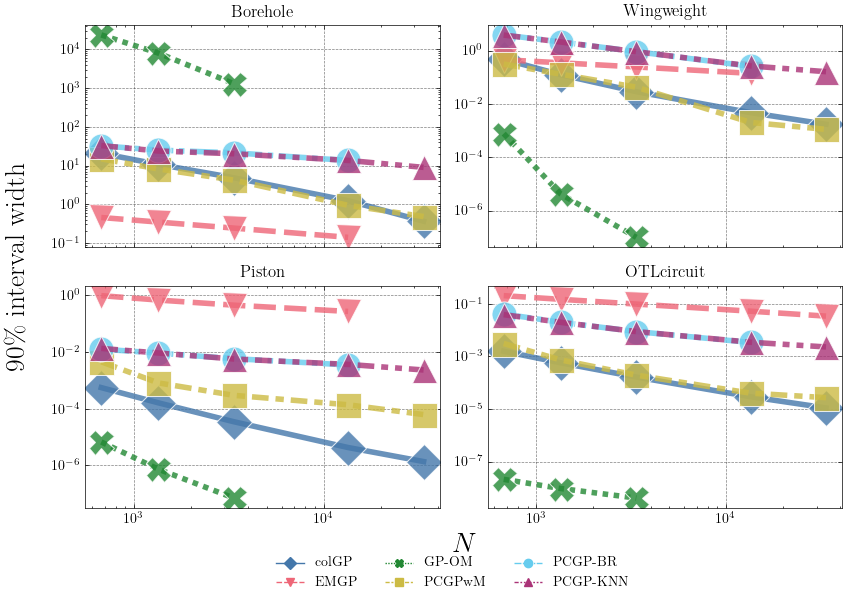}
    \caption{Comparison of 90\% width of surrogate methods, under 25\% MCAR.}
    \label{fig:avgintwidth25random}
\end{figure}

\subsection{Results under MAR}
The results under MAR are presented, in the order of 1\%, 5\%, and 25\%.  
The numerical results under MAR follow similar trends as MCAR for all percentages of missingness.  We refer the analysis to the reporting of results in the MCAR section.  
Figures \ref{fig:rmses1randomMAR}, \ref{fig:rmses5randomMAR}, and \ref{fig:rmses25randomMAR} present the corresponding RMSEs.  Figures \ref{fig:coverages1randomMAR}, \ref{fig:coverages5randomMAR}, and  \ref{fig:coverages25randomMAR} present the 90\% coverages.  And Figures \ref{fig:avgintwidth1randomMAR}, \ref{fig:avgintwidth5randomMAR}, and  \ref{fig:avgintwidth25randomMAR} present the 90\% widths.

\begin{figure}[ht]
    \centering
    \includegraphics[width=0.8\textwidth]{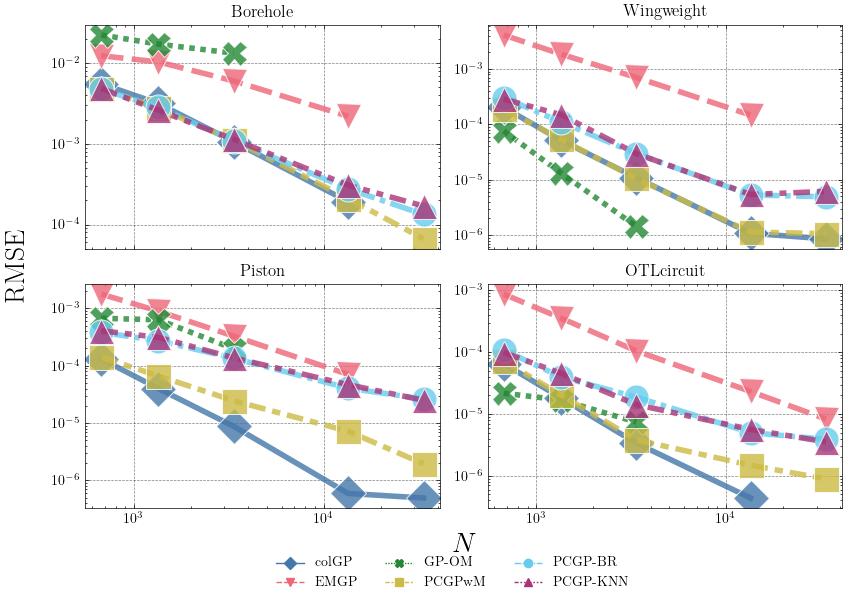}
    \caption{Comparison of prediction accuracy of surrogate methods, under 1\% MAR.}
    \label{fig:rmses1randomMAR}
\end{figure}

\begin{figure}[ht]
    \centering
    \includegraphics[width=0.8\textwidth]{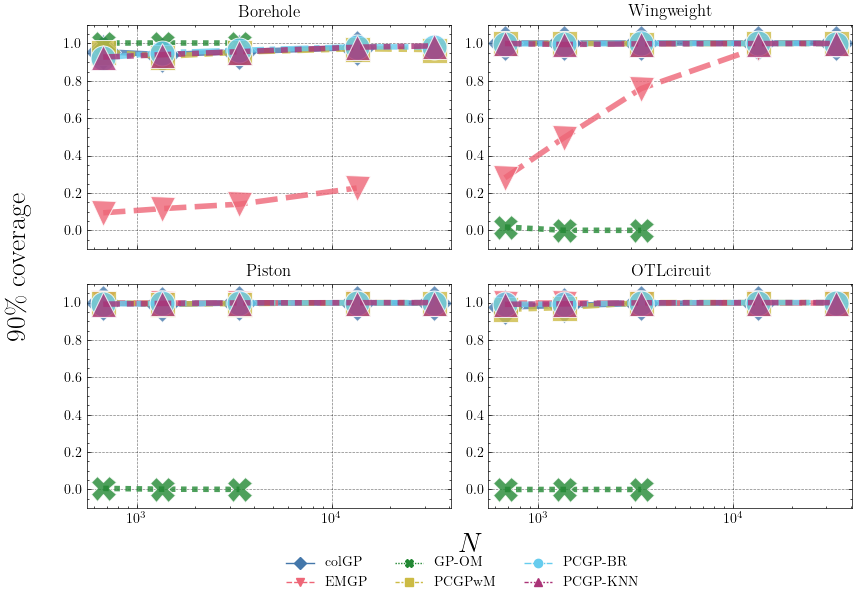}
    \caption{Comparison of 90\% coverage of surrogate methods, under 1\% MAR.}
    \label{fig:coverages1randomMAR}
\end{figure}

\begin{figure}[ht]
    \centering
    \includegraphics[width=0.8\textwidth]{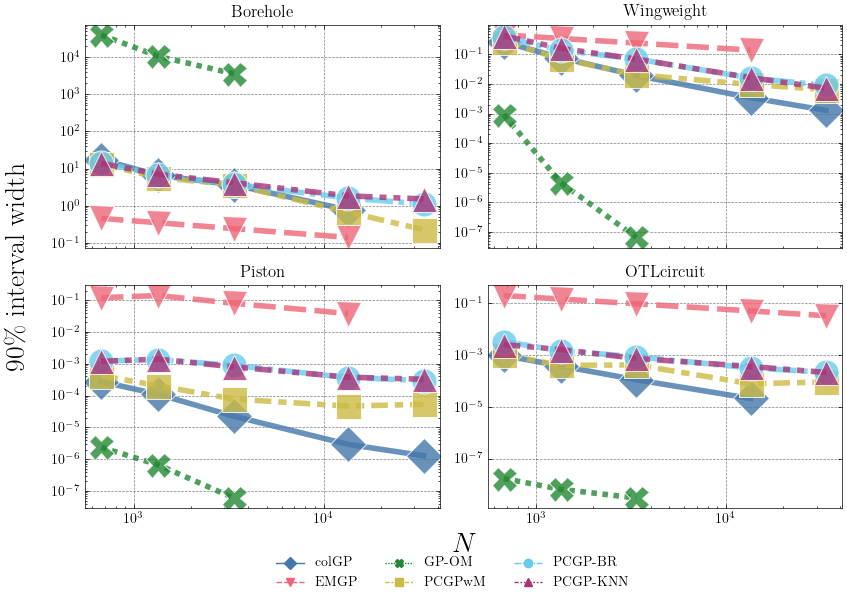}
    \caption{Comparison of 90\% width of surrogate methods, under 1\% MAR.}
    \label{fig:avgintwidth1randomMAR}
\end{figure}

\begin{figure}[ht]
    \centering
    \includegraphics[width=0.8\textwidth]{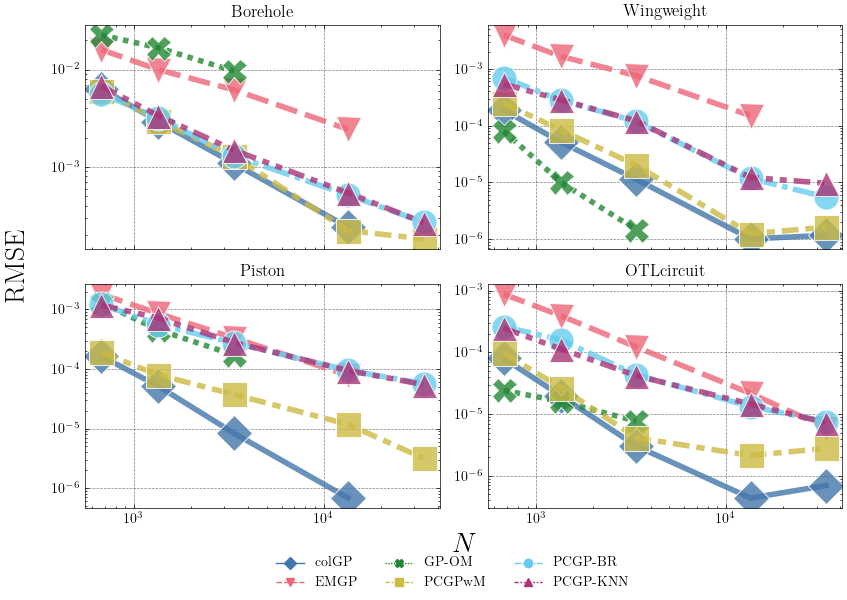}
    \caption{Comparison of prediction accuracy of surrogate methods, under 5\% MAR.}
    \label{fig:rmses5randomMAR}
\end{figure}

\begin{figure}[ht]
    \centering
    \includegraphics[width=0.8\textwidth]{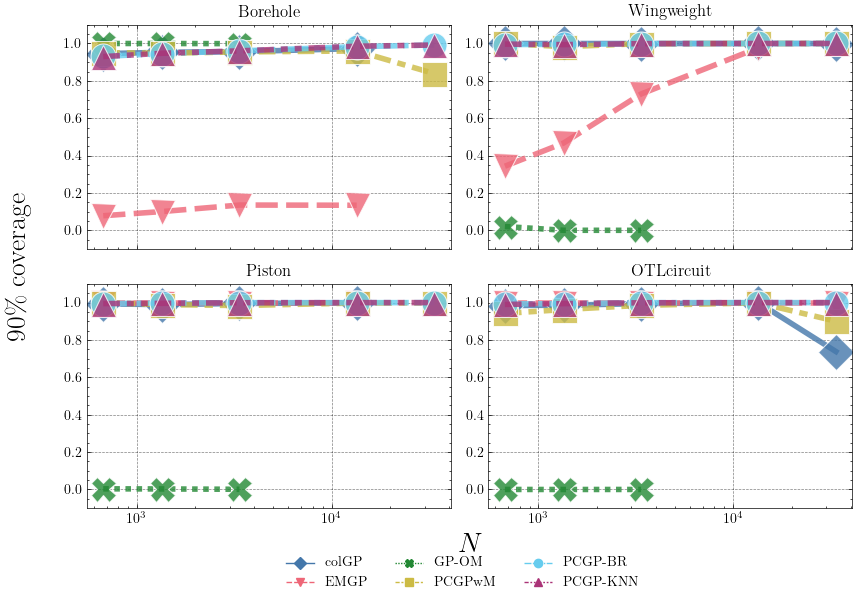}
    \caption{Comparison of 90\% coverage of surrogate methods, under 5\% MAR.}
    \label{fig:coverages5randomMAR}
\end{figure}

\begin{figure}[ht]
    \centering
    \includegraphics[width=0.8\textwidth]{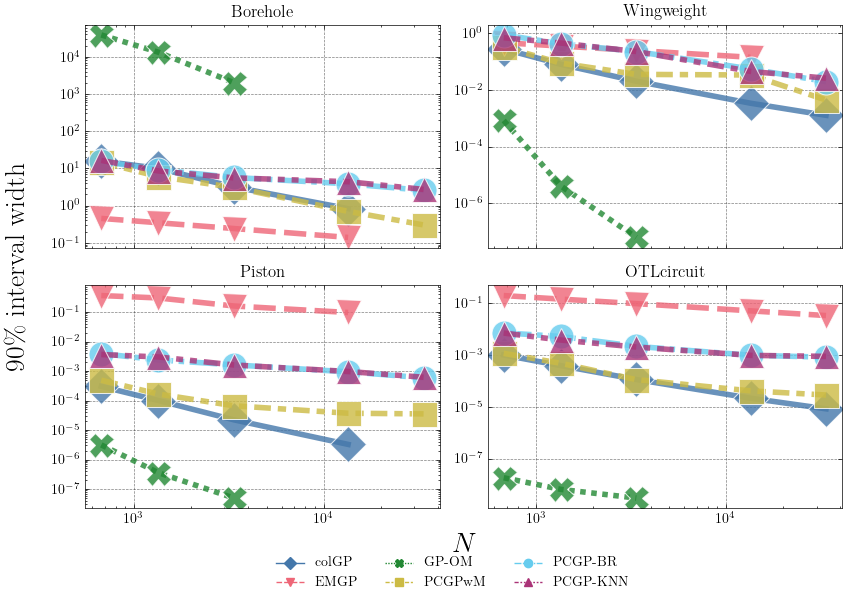}
    \caption{Comparison of 90\% width of surrogate methods, under 5\% MAR.}
    \label{fig:avgintwidth5randomMAR}
\end{figure}

\begin{figure}[ht]
    \centering
    \includegraphics[width=0.8\textwidth]{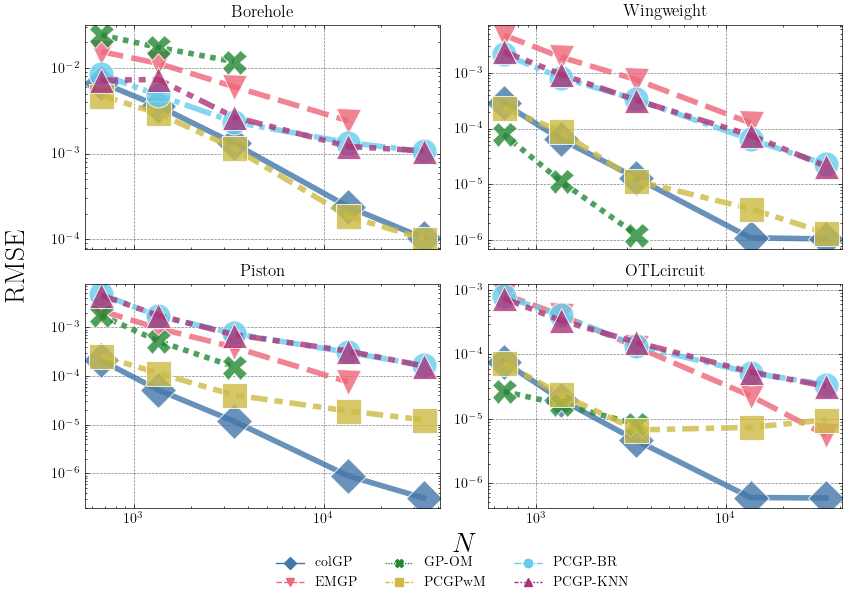}
    \caption{Comparison of prediction accuracy of surrogate methods, under 25\% MAR.}
    \label{fig:rmses25randomMAR}
\end{figure}

\begin{figure}[ht]
    \centering
    \includegraphics[width=0.8\textwidth]{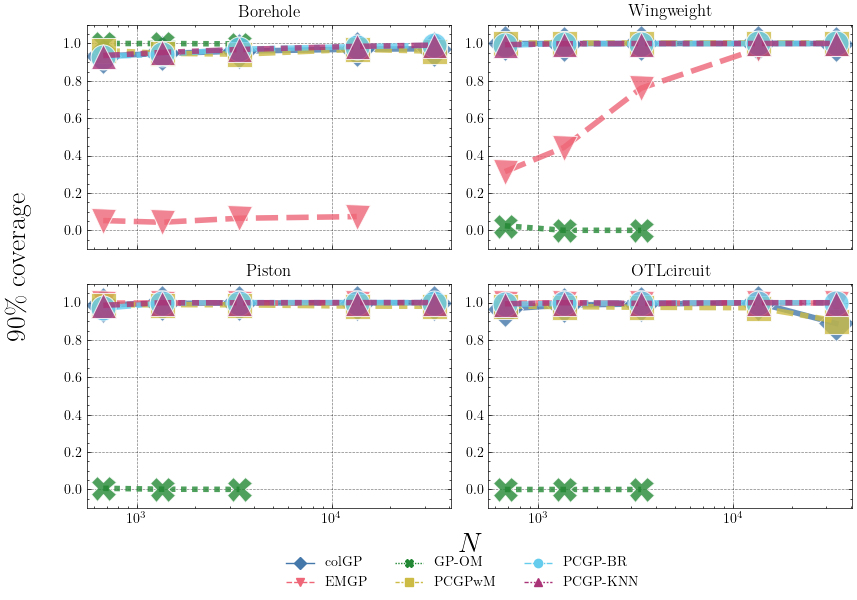}
    \caption{Comparison of 90\% coverage of surrogate methods, under 25\% MAR.}
    \label{fig:coverages25randomMAR}
\end{figure}

\begin{figure}[ht]
    \centering
    \includegraphics[width=0.8\textwidth]{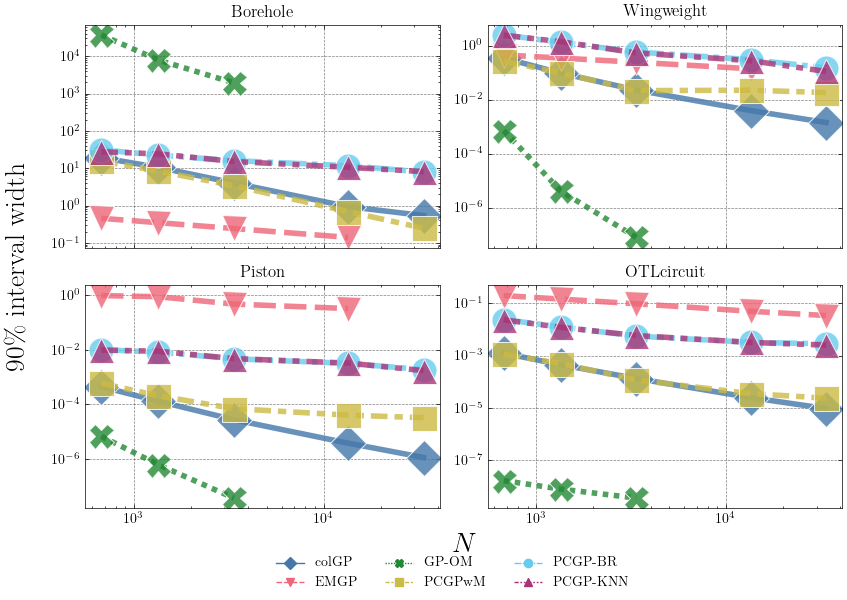}
    \caption{Comparison of 90\% width of surrogate methods, under 25\% MAR.}
    \label{fig:avgintwidth25randomMAR}
\end{figure}

\subsection{Comparison of PCGPwM against a baseline method}
We compare PCGPwM against a baseline method where we build the principal components as if the complete data are provided (i.e., $\bm \Phi$ and $\bm \Lambda$).  We label it the ``PCGP-benchmark'' method.  
The purpose of including this method is to examine the significance of the estimation of the principal components from missing data.
Although improved estimation of principal components from missing data is outside the scope of this article, by studying this method we can offer insights into the limitations of the proposed method.

The two methods are compared under the 5\% MNAR scenario.  The results in terms of RMSE, 90\% coverage, and the 90\% interval width show that PCGPwM performs near identically to PCGP-benchmark, thus the figures are left out.  In other words, the proposed imputation recovers the unknown principal components well.  In certain cases, when compared with colGP (which uses the data without dimension reduction), both PCGPwM and PCGP-benchmark cannot further reduce error even with larger data size.  This points to a potential limitation in the use of principal component methods. 

\section{Scaling of Fayans EDF parameter space}
Table \ref{tbl:fayans_param_space} contains the unscaled centroid, the length scales, and the lower and upper bounds for each dimension of the parameter.  This information is reproduced from Table 5 of \citet{bollapragada2021optimization}.

\begin{table}[ht]
\sisetup{parse-numbers=false, table-format=-3.4}
    \centering
    \begin{tabular}{cSSSS}
    \toprule
    Parameter               & {unscaled center} & {scale} & {lb} & {ub} \\ \midrule 
    $\rho_{\mathrm{eq}}$    & 0.1642 & 0.004    & 0.146     & 0.167 \\ 
    $E/A$                   & -15.86 & 0.1      & -16.21    & -15.50 \\ 
    $K$                     & 206.6 & 25        & 137.2     & 234.4 \\
    $J$                     & 28.3  & 3.2       & 19.5      & 37.0 \\
    $L$                     & 35.9  & 32        & 2.2       & 69.6 \\
    $h^{\mathrm{v}}_{2-}$   & 11.34 & 19.01     & 0         & 100 \\
    $a^{\mathrm{s}}_{+}$    & 0.562 & 0.06      & 0.418     & 0.706 \\
    $h^{\mathrm{s}}_\nabla$ & 0.460 & 0.24      & 0         & 0.516 \\
    $\kappa$                & 0.188 & 0.02      & 0.076     & 0.216 \\
    $\kappa^\prime$         & 0.045 & 0.17      & -0.892    & 0.982 \\
    $f^\xi_\mathrm{ex}$     & -4.46 & 1.16      & -4.62     & -4.38 \\
    $h^\xi_\nabla$          & 4.18  & 1.68      & 3.94      & 4.27 \\
    $h^\xi_+$               & 3.44  & 1.4       & -0.96     & 3.66 \\
    \bottomrule
    \end{tabular}
    \caption{Fayans EDF model parameters and their scaling information, reproduced from Table 5 of \citet{bollapragada2021optimization}.}
    \label{tbl:fayans_param_space}
\end{table}

\section{Code and data for the Fayans EDF case study}
The code and data for the case study are made available in the online supplementary materials.